\documentclass[12pt]{article}

\usepackage{amsthm,amsfonts,amsmath,braket,algorithm,graphicx}
\usepackage{fullpage,authblk}

\usepackage{wrapfig}

\theoremstyle{definition} 
\theoremstyle{definition} 
\newtheorem {theorem} {Theorem}[section]
\newtheorem {corollary} {Corollary}[section]

\newtheorem {lemma} {Lemma}[section]

\newcommand{\kb}[1]{\ket{#1}\bra{#1}}

\newcommand{\N}{\mathbb{N}}

\newcommand{\al}{\mathcal{A}}

\newcommand{\trd}[1]{\left|\left| #1 \right| \right|}

\newcommand{\samp}{\Psi}

\newcommand{\st}{\text{ } | \text{ }}

\newcommand{\Hmin}{H_{\text{min}}}
\newcommand{\Hextd}{\bar{H}}

\newcommand{\alphabet}{\mathcal{A}}

\newcommand{\hd}{\Delta_H}

\newcommand{\leakec}{\texttt{leak}_{\texttt{EC}}}

\newcommand{\experiment}[1]{\leftarrow\textbf{\texttt{Exp}}\left(#1\right)}

\floatname{algorithm}{Protocol}

\title{Quantum Sampling for Optimistic Finite Key Rates in High Dimensional Quantum Cryptography}
\date{}

\author[1]{Keegan Yao}
\author[1]{Walter O. Krawec\footnote{Email: \texttt{walter.krawec@gmail.com}}}
\author[1]{Jiadong Zhu}
\affil[1]{\small{Department of Computer Science and Engineering}\\\small{University of Connecticut}\\\small{Storrs, CT 06269 USA}}

\begin{document}
\maketitle
\begin{abstract}
It has been shown recently that the framework of quantum sampling, as introduced by Bouman and Fehr, can lead to new entropic uncertainty relations highly applicable to finite-key cryptographic analyses. Here we revisit these so-called sampling-based entropic uncertainty relations, deriving newer, more powerful, relations and applying them to source-independent quantum random number generators and high-dimensional quantum key distribution protocols.  Along the way, we prove several interesting results in the asymptotic case for our entropic uncertainty relations. These sampling-based approaches to entropic uncertainty, and their application to quantum cryptography, hold great potential for deriving proofs of security for quantum cryptographic systems, and the approaches we use here may be applicable to an even wider range of scenarios.
\end{abstract}

\section{Introduction}

Quantum sampling, as introduced by Bouman and Fehr in \cite{bouman2010sampling}, is a framework allowing for the analysis of quantum systems through classical statistical sampling methods. Informally, it was shown that when sampling a quantum state (via measuring some subset of it in a particular basis), the remaining, unmeasured, portion of the state behaves like a superposition of states that are ``close'' (with respect to some target value such as Hamming weight) to the observed sample. How close they are depends, in fact, on the error probability of the classical sampling protocol used (where the classical sampling strategy would observe a portion of a classical word in some alphabet and argue about how the remaining, unobserved, portion of the word looks). At a high level, suppose one measures a random portion of some quantum state $\ket{\psi}$ in the $Z = \{\ket{0}, \cdots, \ket{d-1}\}$ basis and always observes $\ket{0}$. Then, one would expect that the remainder of the state (the unmeasured portion) should be a superposition of states that are relatively close to the all $\ket{0\cdots 0}$ state. Bouman and Fehr's framework formalizes this notion, even when the state is entangled with an environment system (e.g., an adversary).

Besides being fascinating on its own, there are now several interesting applications of this work. In their original paper \cite{bouman2010sampling}, the authors showed some applications to quantum cryptography, namely a security proof of the entanglement-based BB84 QKD protocol for qubits (dimension two systems). Recently in \cite{krawec2019quantum,krawec2020new}, we showed how the quantum sampling framework may be used to derive novel quantum entropic uncertainty relations which are highly applicable to finite-key quantum cryptographic security analyses. Informally, quantum entropic uncertainty relations bound the amount of uncertainty in two different measurement outcomes performed on some quantum system. For instance, the famous Maassen and Uffink relation \cite{maassen1988generalized} (which, itself, followed from a conjecture by Kraus in \cite{kraus-uncertainty} and was an improvement over an uncertainty relation proposed first by Deutsch \cite{deutsch-first-bound}) states that, given a quantum state $\rho$ acting on a $d$-dimensional Hilbert space $\mathcal{H}_d$, then if two measurements are performed on the system resulting in random variables $M$ and $N$ respectively, it holds that $H(M) + H(N) \ge \gamma$, where $\gamma$ is a function of the two measurements performed (namely their overlap, though we will formally define this later for our applications). In particular, one cannot in general be certain of the outcome of both measurements of the system. By now there are numerous quantum entropic uncertainty relations with various fascinating properties and applications; for a general survey, the reader is referred to \cite{survey-2,survey,survey-3}.

The so-called \emph{sampling-based entropic uncertainty relations} we introduced in our earlier work \cite{krawec2019quantum,krawec2020new} turn out to be highly useful in finding optimistic secure bit generation rates for quantum random number generation (QRNG) protocols in the source-independent security model \cite{vallone2014quantum}. Our relations bounded the quantum min-entropy $\Hmin(A|E)$ as a function of the Shannon entropy of a particular measurement outcome and the measurement overlap. Since min entropy is a highly valuable resource in quantum cryptography (in particular, it can be used to determine how many uniform random bits one may extract from a source, independent of any adversary \cite{renner2008security}), finding tight bounds on this quantity is highly desirable when analyzing quantum cryptographic protocols. As we've shown in our earlier work, our relations often out-perform prior work in cryptographic settings, producing more optimistic bit generation rates for QRNG protocols leading, potentially, to more rapid implementations of such systems (though here, and in our prior work, we focus only on theoretical analyses - practical settings, though interesting, are outside the scope of this current work). Furthermore, our sampling-based relations incorporate all needed finite sampling effects thus making them easy to use ``out of the box.''

Here, we revisit sampling-based entropic uncertainty relations. These relations involve a quantum state $\rho$, possibly entangled with an adversary, whereby a random sample is chosen and a test is performed by measuring a portion of $\rho$ resulting in some outcome $q$. In this work, we show a highly general, two-party entropic uncertainty relation (Theorem \ref{thm:twoparty}) which, informally, states that with high probability (based on the failure probability of a classical sampling strategy):
\begin{equation}
    \Hmin^\epsilon(A|E) + \log_2 |J_q| \ge n\gamma,
\end{equation}
where $J_q$ is the set of all words in some alphabet that are ``close'' to the observed string $q$; $n$ is the number of qudits that were not measured in the test state; and $\gamma$ is a function of the overlap between the two measurements. One of the strong advantages to our new sampling-based relation is that one may design classical sampling strategies suitable to a quantum cryptographic purpose and simply insert it directly into the above; all one needs to do is analyze the classical error probability and bound or evaluate the size of the set $J_q$ (which is typically a combinatorial proof). Though this result is more general than our original, it turns out the proof of this is nearly identical to our prior work in \cite{krawec2019quantum,krawec2020new}. However the novelty is, first, in the generality of the result that it works for any classical sampling strategy (whereas in \cite{krawec2020new} only a particular sampling strategy was proven); second in its applications, we show that this new bound is powerful enough to analyze a particular source-independent (a form of partial device independence introduced first in \cite{vallone2014quantum}) QRNG protocol producing more optimistic bit-generation rates than prior work using alternative entropic uncertainty relations and, furthermore, unlike our previous work, can provide an alternative proof of the previously mentioned Maassen-Uffink relation for dimensions strictly greater than $2$ (in \cite{krawec2019quantum} we showed this for dimension $2$ systems only).

Our second main contribution is to show a novel three-party sampling-based entropic uncertainty relation involving Alice, Bob, and Eve. Here, Alice and Bob perform a test measurement on some portion of their shared quantum state, resulting in outcome $q_A$ and $q_B$ respectively (these are words in some $d$-character alphabet). Then, informally, our new entropic uncertainty relation (Theorem \ref{thm:three-party}) states that, with high probability:
\begin{equation}
    \Hmin^{\epsilon}(A|E) + \eta_dH_d\left[ \hd(q_A,q_B)+\delta\right] \ge n_0\gamma + n_1\hat{\gamma},
\end{equation}
where $n_0+n_1=n$, the number of systems not measured initially; $\eta_d$ is a constant depending on the dimension ($d$) of the individual systems measured; $\delta$ takes into account imperfect, finite samples; $H_d$ is the $d$-ary Shannon entropy; and $\hd(x,y)$ is the Hamming distance of words $x$ and $y$. Our entropic uncertainty relation can actually incorporate the maximal measurement overlap $\hat{\gamma}$ and the second-maximal overlap $\gamma$, making it useful if the two measurement bases have a similar basis element (e.g., a ``vacuum'' element, useful in QKD when considering channel loss).  This ability shows the great promise in using the Quantum Sampling framework of Bouman and Fehr, augmented with our proof techniques developed here and in our prior work \cite{krawec2019quantum,krawec2020new} to prove interesting, and useful, entropic uncertainty relations.  Indeed, our proof method can even be extended to support additional measurement overlap quantities.

Note that, if $q_A = q_B$, then our result shows that the min-entropy conditioned on the adversary's system $E$ must be high. We use our entropic uncertainty relation to provide a proof of security, in the finite key setting, of the High-Dimensional BB84 protocol \cite{HD-BB84-4,HD-BB84,HD-BB84-2,HD-BB84-3}. Our security proof is valid against arbitrary attacks by an adversary and applies easily to any dimension $d$ of the signal states and can even take into account lossy channels.  Since high-dimensional QKD protocols exhibit many fascinating and useful properties (such as increased noise tolerance \cite{HD-BB84,HD-qkd-survey}), and are experimentally feasible today \cite{HD-exp1,HD-exp2,HD-exp3,HD-exp4}, our new analysis may provide even further benefits to these systems. We note that in \cite{bouman2010sampling}, the sampling framework was used to provide a proof of security for the standard (qubit-based) BB84 using alternative methods which were specific to the qubit-BB84 protocol. Our method provides, first, a novel entropic uncertainty relation which may have numerous other applications to quantum cryptographic protocols outside of HD-BB84; and, secondly, provides as an application a simple proof of security for the high-dimensional variant of BB84 for any dimension $d$ of the system.

This work makes several contributions, not the least of which is showing yet further fascinating, and highly applicable, connections between the quantum sampling framework of Bouman and Fehr \cite{bouman2010sampling} and quantum information theory, in particular entropic uncertainty. Furthermore, our relations are immediately applicable to quantum cryptography in the finite key setting, leading to composable security \cite{renner2008security} and, as we show, in most typical scenarios also highly optimistic secure bit-generation rates for source-independent QRNG protocols and QKD protocols. In practice, such sampling-based approaches show that quantum communication systems may run at higher bit-generation rates than previously thought. Thus, not only does this work provide interesting theoretical contributions, but also potential practical ones (though, as stated, we are not considering practical experimental imperfections here, leaving this as interesting future work). We suspect that there are even more connections and applications of the quantum sampling framework which may shed further light on problems in general information theory and applied quantum cryptography. This paper attempts to take a step forward in that direction.

\subsection{Notation} \label{section:notation}

We start with some notation and definitions that we will use throughout this work. An \textit{alphabet} $\al_d$ is a set of $d$ characters which we typically label $\{0, 1, \cdots, d-1\}$. Given a word $q \in \al_d^n$, the \textit{substring $q_t$ indexed by $t \subset \{1,\dots,n\}$} is the string $q_t = q_{t_1}q_{t_2}\dots q_{t_{|t|}}$. The substring $q_{-t}$ denotes the substring indexed by the complement of $t$.

Much of our work involves arguing about the properties of a given word. In particular, given a string $q \in \al_d^n$, the \textit{relative Hamming weight} is defined as $w(q) = \frac{|\{j \st q_j \neq 0\}|}{n}$ and the \textit{relative character count with respect to $i \in \al_d$} is defined as $c_i(q) = \frac{|\{j \st q_j = i\}|}{n}$. Note that $w(q) = 1-c_0(x)$. We will use $c(q)$ to denote the $d$-tuple of all relative counts, namely $c(q) = (c_0(q), \cdots, c_{d-1}(q))$. The \textit{Hamming distance} between two strings $x,y \in \al_d^n$ is $\hd(x,y) =\frac{|\{i \st x_i \neq y_i\}|}{n}$.

A \textit{density operator} $\rho$ is a positive semi-definite Hermitian operator with trace equal to one, acting on some Hilbert space $\mathcal{H}$. If $\rho_{AE}$ acts on some Hilbert space $\mathcal{H}_A\otimes\mathcal{H}_E$, we write $\rho_E$ to mean the partial trace of $\rho_{AE}$ over $A$ (similarly for other systems).

We use $\mathcal{H}_d$ to denote a $d$-dimensional Hilbert space. Given a basis $\{\ket{v_0}, \cdots, \ket{v_{d-1}}\}$ of $\mathcal{H}_d$, and given a word $i \in \al_d^n$, we write $\ket{v_i}$ to mean $\ket{v_{i_1}}\otimes\cdots\otimes\ket{v_{i_n}}$. If the basis under consideration is clear, we will sometimes write $\ket{i}$ to mean $\ket{v_i}$.

The \textit{Shannon entropy} of a random variable $X$ is denoted by $H(X)$. The \textit{$d$-ary entropy} function $H_d$ is defined as $H_d(x) = d\log_d(d-1) - x\log_d x - (1-x)\log_d(1-x)$. Note that when $d=2$ this is simply the binary Shannon entropy. Finally, we define the \textit{extended $d$-ary entropy} $\Hextd_d(x)$ to be $H_d(x)$ if $0\le x\le 1-1/d$; otherwise $\Hextd_d(x) = 0$ if $x < 0$ or $\Hextd_d(x) = 1$ if $x > 1-1/d$.

Given $\rho_{AE}$ acting on $\mathcal{H}_A \otimes \mathcal{H}_E$, then the \textit{conditional quantum min entropy} \cite{renner2008security} is defined to be:
\begin{equation}
    \Hmin(A|E)_\rho = \sup_{\sigma_E} \max\{\lambda \in \mathbb{R} \st 2^{-\lambda}I_A \otimes \sigma_E - \rho_{AE} \geq 0\}.
\end{equation}
When the $E$ system is trivial, we have $\Hmin(A|E)_\rho = \Hmin(A)_\rho = -\log \max \lambda$, where the maximum is taken over all eigenvalues $\lambda$ of $\rho_A$. In particular, if $\rho_A$ is a classical system (that is, $\rho_A = \sum_ap_a\kb{a}$), then $\Hmin(A)_\rho = -\log\max p_a$. Note that, for any quantum-quantum-classical state $\rho_{AEC} = \sum_{c=0}^N p_c \rho_{AE}^{(c)} \otimes \kb{c}$, then it is easy to prove from the definition of min entropy that the following holds:
\begin{equation}\label{eq:mixed}
    \Hmin(A|EC)_\rho \geq \min_c \Hmin(A|E)_{\rho^{(c)}}.
\end{equation}

Though we will not need it here, a useful interpretation of $\Hmin(A|E)$ for \emph{classical-quantum states (cq-states)} $\rho_{AE}$ (that is, states of the form $\rho_{AE} = \sum_ap_a\kb{a}\otimes\rho_E^{(a)}$) was given in \cite{konig2009operational} as:
\[
\Hmin(A|E)_\rho = -\log P_g(\rho_{AE}),
\]
where $P_g(\rho_{AE})$ is the maximal guessing probability that Eve can guess the value of Alice's register, namely:
\[
P_g(\rho_{AE}) = \max_{\{M_a\}}\sum_ap_atr\left(M_a\rho_{E}^{(a)}\right),
\]
where the maximum is over all POVM operators on $\mathcal{H}_E$.

Finally, the \textit{conditional smooth min entropy} is defined to be \cite{renner2008security}
\begin{equation}
    \Hmin^\epsilon(A|E)_\rho = \sup_{\sigma\in\Gamma_\epsilon(\rho)} \Hmin(A|E)_\sigma.
\end{equation}
where $\Gamma_\epsilon(\rho) = \{\sigma \st \trd{\rho-\sigma}\le \epsilon\}$ and here $\trd{X}$ is the \emph{trace distance} of operator $X$.

For additional notation, given a quantum state $\rho_{AE}$ and an orthonormal basis $Z$ of the $A$ register, we write $\Hmin(Z|E)_\rho$ to mean the conditional min entropy of $\rho_{AE}$ after measuring the $A$ system using the $Z$ basis. If the state $\rho_{AE}$ is pure, namely $\rho_{AE} = \kb{\psi}_{AE}$, we write $\Hmin(A|E)_\psi$. This notation is similar for smooth min entropy.

The following Lemma relating the min entropies of mixed and pure states will be useful to our work later as it will allow us to bound the min entropy of a superposition of states by, instead, computing the min entropy of a corresponding mixture of states:

\begin{lemma}\label{lemma:superposition} (From \cite{bouman2010sampling} based also on a Lemma in \cite{renner2008security}) Let $Z=\{\ket{i}\}$ and $X=\{\ket{x_i}\}$ be two orthonormal bases of $\mathcal{H}_A$. Then for any pure state $\ket{\psi} = \sum_{i\in J} \alpha_i \ket{i} \otimes \ket{\phi_i}_E \in \mathcal{H}_A \otimes \mathcal{H}_E$ (where $\ket{\phi_i}_E$ are arbitrary, normalized states in $\mathcal{H}_E$), if we define the mixed state $\rho = \sum_{i\in J} |\alpha_i|^2 \kb{i} \otimes \kb{\phi_i}$, then
\[\Hmin(X|E)_\psi \geq \Hmin(X|E)_\rho - \log_2|J|.\]

\end{lemma}

Quantum min entropy is of vital importance to quantum cryptography as it allows one to determine how many uniform random bits one may extract from a $cq$-state $\rho_{AE}$ that are also independent of Eve. In particular, given a $cq$-state (which, itself, is typically the result of running some quantum cryptographic protocol where the $A$ register may not be uniform random or completely independent of the $E$ register), one may apply the process of \emph{privacy amplification} (typically running the $A$ register through a randomly chosen two-universal hash function) to establish the required uniform and independent random string. If $\sigma_{KE}$ is the result of applying privacy amplification to the initial $\rho_{AE}$ system, where the $K$ register is of size $\ell$ bits, it was shown in \cite{renner2008security} that:
\begin{equation}\label{eq:PA}
\trd{\sigma_{KE} - \frac{I_K}{2^\ell}\otimes\sigma_E} \le \sqrt{2^{(\Hmin^\epsilon(A|E)_\rho - \ell)}} + 2\epsilon.
\end{equation}
Thus, by deriving a lower-bound on the min entropy of the initial state $\rho_{AE}$ before privacy amplification, one may establish how many uniform and independent bits may be extracted (namely, $\ell$) from the state to satisfy the above trace distance inequality up to a desired level of security; e.g., so that the difference between the real state $\sigma_{KE}$ and the ``ideal'' state $I_K/2^\ell\otimes\sigma_E$ (which represents a uniform random string, independent of any other system) is no more than some $\epsilon_{PA}$.



\section{Quantum Sampling}

In~\cite{bouman2010sampling}, Bouman and Fehr discovered a fascinating connection between classical sampling strategies and quantum sampling. Since our work utilizes this as a foundation to prove our entropic uncertainty relations (later used to prove security of QRNG and QKD protocols), we take the time in this section to provide a review of their main results. Everything in this section, definitions, concepts, and theorems, come from \cite{bouman2010sampling} except when explicitly mentioned. Occasionally, we will make some generalizations and simplifications, however wherever we do so, it will be made clear in the narrative.

Let $\alphabet_d$ be an alphabet with $d$ characters and $N \in \mathbb{N}$ be fixed. A \emph{classical sampling strategy} is a triple $\Psi=(P_T, P_S, f)$, where $P_T$ is a probability distribution over subsets of $\{1, 2, \cdots, N\}$, $P_S$ is a probability distribution over some set $\{0,1\}^*$ called \emph{seed values}, and $f$ is a function:
\begin{equation}
    f : \{0,1\}^* \times \alphabet_d^* \rightarrow \mathbb{R}^k.
\end{equation}
Given a string $q \in \alphabet^N$, the strategy consists of, first, sampling a subset $t$ according to $P_T$; sampling a seed value $s$ according to $P_S$, observing the value of $q_t$ and evaluating $f(s, q_t)$. This evaluation should lead to a ``guess'' of the value of some target function $g:\alphabet_d^*\rightarrow \mathbb{R}^k$ evaluated on the \emph{unobserved} portion of $q$, namely $q_{-t}$. Informally, a good sampling strategy will ensure that, with high probability, $\max_i|f_i(s, q_t) - g_i(q_{-t})| \le \delta$ (i.e., the difference in all coordinates of the output function evaluated on the sampled portion of $q$, compared to the target function evaluated on the unobserved portion, are no greater than $\delta$). Note that above, we are generalizing the sampling result of \cite{bouman2010sampling} to include more general target and guess functions; in \cite{bouman2010sampling}, $k=1$ and $g(x) = w(x)$, the Hamming weight of $x$. However, the proof of their main result is easily seen to hold in this more general case, so long as suitable \emph{classical} strategies are analyzed appropriately (as we do later in this section). Finally, note that in our work, we do not make use of this additional random seed value (which is useful when implementing randomized guess functions $f$); thus, we disregard writing it from here on out and, instead, our function $f$ simply maps strings from $\alphabet_d^{|t|}$ to values in $\mathbb{R}^k$.

Now, fix a subset $t \subset \{1, 2, \cdots, N\}$ and $\delta \ge 0$ and consider the set:
\begin{equation}
\mathcal{G}_{t,\delta}^{f,g} = \mathcal{G}_{t,\delta} = \{i \in \alphabet_d^N \st \max_j |f_j(i_t) - g_j(i_{-t})| \le \delta\}.
\end{equation}

This set consists of all ``good'' words in $\alphabet_d^N$ where, for the given choice of $t$, the estimate produced by $f$ is $\delta$ close to the desired target function on the unobserved portion. Note that, when the context is clear, we will forgo writing the $f$ and $g$ superscripts. From this, the \emph{error probability} of the given classical sampling strategy is defined to be:
\begin{equation}\label{eq:error-prob}
    \epsilon_\delta^{cl}(\Psi) = \max_{q\in\alphabet_d^N}Pr\left(q \not\in \mathcal{G}_{T,\delta}\right),
\end{equation}
where the probability is over the choice of subsets $t$ drawn according to $P_T$ (the notation $\mathcal{G}_{t,\delta}$ is used to denote the set defined above for a fixed $t$ whereas $\mathcal{G}_{T,\delta}$ denotes a random variable over the choice of subset $t$). Note that the randomness here is only over the choice of subset; if the function $f$ need also make random choices, this could be incorporated through the use of the additional seed value. Since our strategies we use here do not need this, we forgo considering it.

From the above definition, it is clear that for any $q \in \alphabet_d^N$, the probability that the sampling strategy fails to produce an accurate estimate of the target function is at most $\epsilon_\delta^{cl}$. The ``cl'' superscript is used to denote that this is the failure probability of the \emph{classical} sampling strategy.

These notions may be adapted to quantum states. Let $\mathcal{H}_d$ be the $d$-dimensional Hilbert space spanned by some orthonormal basis $\mathcal{B} = \{\ket{0}, \cdots, \ket{d-1}\}$. The choice of basis may be arbitrary, however all following definitions are taken with respect to the chosen basis. 

Given a classical sampling strategy $(P_T, f)$ (again, disregarding the seed $P_S$ which we do not use) and a quantum input state $\ket{\psi} \in \mathcal{H}_d^{\otimes N}\otimes\mathcal{H}_E$, a \emph{quantum sampling strategy} may be constructed as follows: first, sample $t$ according to $P_T$; second, measure those qudits in $\mathcal{H}_d^{\otimes N}$ indexed by $t$ using basis $\mathcal{B}$ to produce measurement result $q_t\in\mathcal{A}_d^{|t|}$; finally, evaluate the function $f(q_t)$. The main result from \cite{bouman2010sampling}, informally, is that the remaining \emph{unmeasured} portion of the input state should behave like a superposition of states that are $\delta$ close in the target function $g(\cdot)$ to the estimated value $f(q_{t})$.

More formally, consider:
\[
span\left(\mathcal{G}_{t,\delta}\right) = span\left\{\ket{b} \st b \in \mathcal{G}_{t,\delta}\right\},
\]
 where, by $\ket{b}$, we mean $\ket{b_1}\otimes\cdots\otimes\ket{b_N}$ (again, with respect to the given basis).
Note that, if $\ket{\psi}_{AE}\in span\left(\mathcal{G}_{t,\delta}\right)\otimes\mathcal{H}_E$, and if subset $t$ is actually the one chosen by the sampling strategy, then it is guaranteed that, after measuring those qudits indexed by $t$ in the given basis $\mathcal{B}$ resulting in outcome $q_t$, the remaining unmeasured portion will be in a superposition of states of the form:
\[
\ket{\psi_q}_{A_{-t}E} = \sum_{i\in J_q}\alpha_i\ket{i}\otimes\ket{E_i},
\]
where:
\[
J_q = \{i \in \alphabet_d^{N-|t|} \st \max_j |f_j(q) - g_j(i)| \le \delta\}.
\]


Formally, the main result from \cite{bouman2010sampling} is stated below, which argues that the input state will be $\epsilon$ close in trace distance to an ideal state where this sampling process always yields the correct guess and this collapse always happens. Furthermore, the $\epsilon$ depends on the error probability of the underlying classical sampling strategy.

\begin{theorem} \label{thm:sample} (From \cite{bouman2010sampling}, though reworded for our application):
Let $\Psi=(P_T, f)$ be a classical sampling strategy with classical failure probability $\epsilon_\delta^{cl}$ for given $\delta > 0$. Then, for every state $\ket{\psi}_{AE} \in \mathcal{H}_{A}\otimes\mathcal{H}_E$ with $\mathcal{H}_A\cong \mathcal{H}_d^{\otimes N}$, there exists a collection of states $\{\ket{\phi^t_{AE}}\}_{t}$ indexed by subsets $t$ of $\{1, \cdots, N\}$ with each $\ket{\phi_{AE}^t}\in span\left(\mathcal{G}_{t,\delta}\right) \otimes \mathcal{H}_E$ such that
\begin{equation}
    \frac{1}{2}\bigg|\bigg|\sum_t P_T(t)\kb{t}\otimes \kb{\psi} -\sum_t P_T(t) \kb{t}\otimes \kb{\phi^t_{AE}} \bigg|\bigg| \leq \sqrt{\epsilon_\delta^{cl}(\Psi)},
\end{equation}
where $t$ represents a sampled subset of $\{1,\dots,N\}$.
\end{theorem}
\begin{proof}
In Bouman and Fehr's work \cite{bouman2010sampling}, it was shown that for a fixed $\ket{\psi}_{AE}$ it holds that
\begin{equation}
\min_{\{\ket{\phi^t_{AE}}\}} \trd{\sum_t P_T(t)\kb{t}\otimes \kb{\psi}_{AE} -\sum_t P_T(t) \kb{t}\otimes \kb{\phi^t_{AE}}} \le \sqrt{\epsilon_\delta^{cl}}
\end{equation}
where the minimum is over all $\{\ket{\phi^t_{AE}}\} \subset span\left(\mathcal{G}_{t,\delta}\right) \otimes \mathcal{H}_E$, for a sampling strategy where the target function was $g(x) = w(x)$. However, in their proof, the above is shown directly by projecting the input $\ket{\psi}_{AE}$ into the space $span\left(\mathcal{G}_{t,\delta}\right) \otimes \mathcal{H}_E$, thus directly constructing the ideal states. Namely, the ideal states were defined by the decomposition $\ket{\psi}_{AE} = \braket{\widetilde{\phi^t_{AE}} | \psi_{AE}} \ket{\widetilde{\phi^t_{AE}}} + \braket{{\phi^t_{AE}} | \psi_{AE}}\ket{\phi^t_{AE}}$ where the $\ket{\widetilde{\phi^t_{AE}}}$ lives in a space orthogonal to the ideal. This minimum is therefore attained by these ideal states. Furthermore, there is no specific reason in this construction to restrict to target functions that are the Hamming weight, nor to target functions that are one-dimensional. Indeed, by considering any definition of $\mathcal{G}_{t,\delta}$, their construction and the subsequent analysis follows identically assuming the error probability is defined as in Equation \ref{eq:error-prob} based on the set $\mathcal{G}_{t,\delta}$. The important difference comes in the analysis of the classical sampling strategy in order to compute $\epsilon_\delta^{cl}$.
\end{proof}

The fascinating thing about Theorem \ref{thm:sample} is that, by choosing suitable classical sampling strategies, one may analyze the behavior of ideal states which always behave appropriately for the given strategy. From this, and the fact that the real state is close, in trace distance, to these ideal states (on average over the randomness in the sampling strategy), one may then promote the analysis from the ideal state to the actual input. Already in \cite{krawec2019quantum,krawec2020new}, we used this to prove novel, and useful, quantum entropic uncertainty relations which were then used to analyze particular QRNG protocols. We now generalize these results, analyze a more powerful QRNG protocol, and also show how this can be used to develop three-party entropic uncertainty relations (involving $A$, $B$, and $E$) with applications to high-dimensional QKD protocols. We show that, furthermore, this provides highly optimistic secure bit generation rates for both the QRNG and QKD protocols in a variety of scenarios. However, to analyze these protocols, we first require some important classical sampling strategies.

\subsection{Classical Sampling Strategies} \label{section:sample-strats}

As discussed, Theorem \ref{thm:sample} allows us to consider classical sampling strategies and use these to analyze quantum protocols. Here we discuss four classical sampling strategies which we denote $\samp_0, \samp_1,$ $\samp_2,$ and $\samp_{2+0}$. Strategy $\samp_0$ was analyzed in \cite{bouman2010sampling} and we use this to bound the error of the other strategies. The other strategies involve one party ($\samp_1$) or two parties ($\samp_2$ and $\samp_{2+0}$) and will be used later when deriving our entropic uncertainty relations.

$ $\newline\textbf{One-Party HD-Restricted-Sampling $\samp_0$:} In \cite{bouman2010sampling}, the following natural sampling strategy was analyzed which we denote here as $\Psi_0$. We use this result to bound the error in our other sampling strategies to be discussed next. Let $q\in\mathcal{A}_d^{n+m}$ be a string and the target function $g(x) = w(x)$. The strategy, first, chooses a subset $t$ of $\{1,\cdots, n+m\}$ of size $m$, uniformly at random and observes string $q_t$. Next, it outputs $f(q_t) = w(q_t)$, an estimate of the Hamming weight of the unobserved portion, namely $w(q_{-t})$. We call this the HD-Restricted-Sampling strategy as it is high-dimensional, however it only looks at the Hamming weight, ignoring the counts of other characters. The following Lemma was proven in \cite{bouman2010sampling}:
\begin{lemma}\label{lemma:samp-psi0}
(From \cite{bouman2010sampling}): Let $\delta > 0$ and $d \ge 2$. Then the failure probability of the above described sampling strategy $\Psi_0$ for $m \le n$ is:
\[
\epsilon_\delta^{cl}(\Psi_0) \le 2\exp\left(\frac{-\delta^2m(n+m)}{m+n+2}\right).
\]
\end{lemma}

We comment that there is nothing special in the above sampling strategy, or their proof, about the use of the Hamming weight in the above Lemma; instead one could replace the target function $g(x)$ with any single $c_j(x)$ or $1-c_j(x)$ (to count the number of letters equal to, or not equal to, $j$ respectively) and the same bound will follow (for a single, fixed but arbitrary, $j$). See \cite{bouman2010sampling}.

$ $\newline\textbf{One-Party HD-Full-Sampling $\samp_1$:} In our work, here, we will need three additional sampling strategies. The first sampling strategy, which we denote $\samp_1$, is a one-party strategy involving Alice only and will be used for our QRNG analysis later. The strategy works for strings in $\alphabet_d^N$, where $N=n+m$ and the target function is $g(x) = (c_0(x),\dots,c_{d-1}(x))$ where $c_i(x)$ is the relative number of times symbol $i$ appears in the word $x$ (as defined in Section \ref{section:notation}). First, the strategy $\samp_1$ chooses a subset $t$ of size $m$ from $\{1, \cdots, N\}$ uniformly at random and observes the string $q_t \in \alphabet_d^{m}$. Finally, $\samp_1$ outputs $f(q_t) = (c_0(q_t),\dots,c_{d-1}(q_t))$ as an estimate of the relative counts of the unobserved $q_{-t}$. The proceeding Lemma determines an upper bound on the error probability of the sampling strategy $\samp_1$.

\begin{lemma}\label{lemma:samp-psi1} Let $\delta > 0$ and $d \geq 2$. Then the failure probability of the above described sampling strategy $\samp_1$ when $m \le n$ is:
\[\epsilon_\delta^{cl}(\samp_1) \leq 2d\exp\left(-m\delta^2\frac{m+n}{m+n+2}\right).\]
\end{lemma}

\begin{proof}
Note that, for any $j$, $(P_T, c_j)$ is exactly the strategy $\samp_0$ (though, instead of looking at the number of strings with a certain Hamming weight, we are looking at the number of strings with a certain character count). Thus, using the bound provided by Lemma \ref{lemma:samp-psi0} we find
\allowdisplaybreaks
\begin{align*}
    \epsilon_\delta^{cl} &= \max_{q\in\alphabet_d^{m+n}}Pr\left(q \not\in \mathcal{G}_{T,\delta}(\Psi_1)\right) \\
    &\le \sum_j \max_{q\in\alphabet_d^{m+n}} Pr\left(|f_j(q_t) - g_j(q_{-t})| > \delta\right) \\
    &\le 2d\exp\left(-m\delta^2\frac{m+n}{m+n+2}\right).
\end{align*}
\end{proof}

$ $\newline\textbf{Two-Party HD-Sampling $\samp_2$:} The second strategy we require will be used for our two-party applications later and we denote by $\samp_2$. Here, we have an input string $q = (q^A, q^B) \in \alphabet_d^N\times \alphabet_d^N$, where $N=n+m$. The strategy will first choose a subset $t \subset\{1, \cdots, N\}$ of size $m$ uniformly at random. The strategy will then sample $q^A_t$ and $q^B_t$; that is, it will observe the $q^A$ portion and $q^B$ portion individually, using the same subset (this may be written strictly using our earlier definitions, however such strict formality is not enlightening). The target function is $g(q^A_{-t},q^B_{-t}) = \hd(q^A_{-t},q^B_{-t})$ (where $\hd(x,y)$ is the relative Hamming distance of words $x$ and $y$ as defined in Section \ref{section:notation}) and the output will be $f(q^A_t, q^B_t) = \hd(q^A_t,q^B_t)$. Again, we may bound the error probability of this strategy using Lemma \ref{lemma:samp-psi0}.

\begin{lemma}\label{lemma:samp-psi2}
Let $\samp_2$ be the strategy defined above; $\delta > 0$ and $m\le n$. Then $\epsilon^{cl}_\delta(\samp_2) \le \epsilon^{cl}_\delta(\Psi_0)$.
\end{lemma}
\begin{proof}
Let $N = n+m$ and $\mathcal{G}_{t,\delta} = \{(i,j) \in \alphabet_d^N\times\alphabet_d^N \st |\hd(i_t,j_t) - \hd(i_{-t},j_{-t})| \le \delta \}$ and $\mathcal{G}_{t,\delta}' = \{i\in\mathcal{A}^N \st |w(i_t) - w(i_{-t})|\le \delta\}.$ Pick $q = (q^A,q^B) \in \alphabet_d^N\times\alphabet_d^N$ and let $x = q^A - q^B$, where the subtraction here is character-wise, modulo $d$, in the given alphabet. Clearly $w(x_t) = \hd(q^A_t, q^B_t)$, and similarly for $x_{-t}$. Thus, $q \in \mathcal{G}_{t,\delta}$ if and only if $x \in \mathcal{G}_{t,\delta}'$. Hence, for every $q=(q^A,q^B)$, it holds that:
\[
Pr\left(q^Aq^B\not\in\mathcal{G}_{T,\delta}\right) = Pr\left(q^A-q^B \not \in \mathcal{G}'_{T,\delta}\right) \le \max_{x\in\alphabet_d^N}Pr\left(x\not\in\mathcal{G}'_{T,\delta}\right) = \epsilon^{cl}_\delta(\Psi_0).
\]
Since this holds for any $q=(q^A,q^B)$, we're done.
\end{proof}

Finally, we define a second two-party sampling strategy which combines $\samp_2$ with $\samp_0$; we denote this strategy by $\samp_{2+0}$.  For this strategy, the target function is now $g(q_{-t}^A, q_{-t}^B) = (\hd(q_{-t}^A,q_{-t}^B), c_{b^*}(q_{-t}^A))$ for some given, fixed, distinguished index $b^* \in \alphabet_d$ (we later call this the ``count index'').  This sampling strategy chooses a subset according to $\samp_2$ and outputs a guess $f(q^A_t,q^B_t) = (\hd(q^A_t,q^B_t), c_{b^*}(q^A_t)).$  It is not difficult to show from Lemmas \ref{lemma:samp-psi0} and \ref{lemma:samp-psi2} that the error probability of this strategy is:
\begin{equation}\label{eq:samp-psi20}
\epsilon_\delta^{cl}(\samp_{2+0}) \le \epsilon_\delta^{cl}(\samp_{2}) + \epsilon_\delta^{cl}(\samp_{0}) \le 4\exp\left(\frac{-\delta^2m(n+m)}{m+n+2}\right).
\end{equation}

\section{Quantum Sampling Based Entropic Uncertainty}

In \cite{krawec2019quantum,krawec2020new}, we showed how the technique of quantum sampling, introduced in \cite{bouman2010sampling} and discussed in the previous section, can be used to prove entropic uncertainty relations bounding the smooth quantum min entropy and the Shannon entropy, as a function of the overlap of two projective measurements. Our first work \cite{krawec2019quantum} introduced a novel entropic uncertainty relation applicable to qubits (i.e., $d=2$) only and with a \emph{fixed} sampling strategy; in \cite{krawec2020new}, we expanded the result to work for qudits ($d\ge 2$), however only with a partial basis measurement and a particular, fixed, sampling strategy. Here, we discuss and generalize this result to work with more general sampling strategies allowing a ``plug-and-play'' entropic uncertainty relation for various classical sampling strategies. Indeed, as shown in this section, one may introduce an arbitrary classical sampling strategy (perhaps one that is useful for a particular cryptographic application); one need only compute the error probability of the given classical strategy, along with the size of a set similar to $\mathcal{G}$ (generally a classical combinatorial proof) to derive a result applicable to a quantum system. The proof of this follows the same two-step approach we introduced in \cite{krawec2019quantum,krawec2020new} only with suitable generalizations at certain points.

To describe our sampling based entropic uncertainty relations, we require an \emph{experiment} which takes as input a quantum state $\rho$ acting on $\mathcal{H}_T\otimes\mathcal{H}_A\otimes\mathcal{H}_E$ where the $A$ portion is an $N$-fold tensor of some smaller $d$-dimensional Hilbert space and the $T$ register is a Hilbert space spanned by orthonormal basis $\{\ket{t}\}$ where $t \subset \{1, \cdots, N\}$. The experiment also requires an orthonormal basis $X = \{\ket{x_0},\cdots\ket{x_{d-1}}\}$.

The experiment will first choose a random subset $t$ by measuring the $T$ register. It will then measure the $A$ portion of $\rho$, indexed by $t$, using the given $X$ basis. This measurement results in outcome $q \in \al_d^{|t|}$ and a post-measurement state $\rho(q,t)$, acting on the unmeasured portion of $\mathcal{H}_A$ and $\mathcal{H}_E$. We denote this experiment by $(t,q, \rho_{A'E}(q,t))\experiment{\rho_{TAE}, X}$. Note that the experiment also returns the subset chosen. Sampling based entropic uncertainty relations allow one to bound the min entropy in the remaining post-measured state, assuming an alternative measurement were to be made on the $A$ portion of it. This bound is a function of the measurement overlap and the classical measurement outcome $q$.

The main result from \cite{krawec2019quantum,krawec2020new} was to relate the min entropy in the remaining portion of the system as a function of the measurement overlap and the binary Shannon entropy (or, in the case of \cite{krawec2020new}, the $d$-ary Shannon entropy) of the relative Hamming weight of the observed outcome $q$ after running the experiment. However, the proof technique used there can be applied to a more general setting allowing for arbitrary sampling strategies and, in particular, to bound the min-entropy as a function of the measurement overlap and the size of a particular set $J_q$ of classical strings that are $\delta$-close to the observed $q$.

\begin{theorem} \label{thm:twoparty}
Let $0 < \beta < 1/2$ and $\samp$ be a classical sampling strategy with error probability $\epsilon_\delta^{cl}$ for given $\delta > 0$. Let $\epsilon = \sqrt{\epsilon^{cl}_\delta}$, and let $\rho_{AE}$ be an arbitrary quantum state acting on space $\mathcal{H}_A \otimes \mathcal{H}_E$, where $\mathcal{H}_A \cong \mathcal{H}_d^{\otimes N}$ for $d \geq 2$. Let $Z = \{\ket{z_i}\}_{i=0}^{d-1}$ and $X = \{\ket{x_i}\}_{i=0}^{d-1}$ be two orthonormal bases of $\mathcal{H}_d$. Furthermore, let $(t, q, \rho(t, q)) \leftarrow \textbf{Exp}(\sum_tP_T(t)\kb{t}\otimes \rho_{AE}, X)$, where the sum is over all possible subsets of $\{1, 2, \dots, N\}$ that could be chosen by $\samp$ and $P_T(t)$ is the probability of subset $t$ being chosen as determined by the given classical sampling strategy. Finally, let $\gamma = -\log_2\max_{a,b}|\braket{z_a|x_b}|^2$.  Then, it holds that:
\begin{equation}
    Pr\left(\Hmin^{4\epsilon+2\epsilon^\beta}(Z|E)_{\rho(t,q)} + \log_2 |J_{q}^{(N-|t|)}| \geq (N-|t|)\gamma\right) \geq 1-2\epsilon^{1-2\beta},
\end{equation}
where
\begin{equation}\label{eq:Jq}
    J_{q}^{(n)} = \{i\in\alphabet_d^{n} \st \max_j|f_j(i) - g_j(q)| \le \delta\}.
\end{equation}
Above the probability is over the randomness in the experiment (namely the subset chosen and the resulting measurement outcome $q$).
\end{theorem}

\begin{proof}

The proof follows the same two-step argument we developed in \cite{krawec2019quantum,krawec2020new}. In fact, most of the proof is identical with the exception of a few generalizations; we provide the proof here at a high-level only for completeness, referring the reader to \cite{krawec2019quantum,krawec2020new} for complete technical details when needed.

$ $\newline\textbf{First Step - Ideal Analysis:} We begin by considering the case when the input state $\rho_{AE}$ is pure; the mixed case then follows through standard purification techniques.

By applying Theorem \ref{thm:sample} with respect to the given $X$ basis and sampling strategy $\samp$, there exist ideal states $\{\ket{\phi_{AE}^t}\}$ such that for every $t$, the state $\ket{\phi^t_{AE}} \in \text{span}\{\ket{x_i} \st i \in \mathcal{A}_d^N \text{ and } \max_j |f_j(i_t) - g_j(i_t)| \le \delta\}\otimes\mathcal{H}_E$. Note that the target function $g(x) = (g_1(x), \cdots, g_k(x))$ also depends on the sampling strategy. Furthermore, from this application of Theorem \ref{thm:sample}, if we define $\sigma_{TAE} = \sum_t P_T(t)\kb{t}\otimes\kb{\phi_{AE}^t},$ then it holds that:
\begin{equation}
    \trd{\sum_t P_T(t)\kb{t}\otimes \rho_{AE} - \sigma_{TAE}} \le \sqrt{\epsilon^{cl}_\delta(\samp)} = \epsilon.
\end{equation}
Consider the output of running $(t,q,\sigma(t,q))\experiment{\sigma, X}$. Here $q \in \alphabet_{d}^{|t|}$. It is not difficult to see that the resulting state, after tracing out the measured portion, is of the form:
\begin{equation}
 \sigma(t,q) = \sum_{i\in J_{q}^{(N-|t|)}} \alpha_{i} \ket{x_i}\otimes\ket{E_i},
\end{equation}
where $J_q^{(n)} = \{i \in \alphabet_d^n \st \max_j|f_j(i) - g_j(q)|\le \delta\}$ (note that some of the $\alpha_i$'s may be zero).

Let $n = N-|t|$. 
From Lemma \ref{lemma:superposition}, we have $\Hmin(Z|E)_{\sigma(t,q)} \ge H(Z|E)_\chi - \log|J_q^{(n)}|$, where $\chi$ is the mixed state:
\[
\chi_{AE} = \sum_{i\in J_q^{(n)}}|\alpha_i|^2\kb{x_i}\otimes\kb{E_i}.
\]

It is straight-forward to show that $\Hmin(Z|E)_\chi = n\gamma = (N-|t|)\gamma$. This is done by conditioning on an additional classical system, writing out the probability distribution of the $Z$ basis measurement given $\chi$ and taking advantage of Equation \ref{eq:mixed} (see \cite{krawec2020new} for explicit details on how this computation is done given a mixed state of this form). Thus, \emph{with certainty}, the ideal case, after choosing subset $t$ and observing $q$, will have min entropy no less than $(N-|t|)\gamma - \log|J_q^{(n)}|$.

$ $\newline\textbf{Second Step - Real Case Analysis:} The second step involves arguing that the real state cannot behave too differently from the ideal state we just analyzed. We make use of Chebyshev's inequality while also switching to smooth min entropy to complete the analysis.

Consider the real state $\rho= \frac{1}{T}\sum_t\kb{t}\otimes\rho_{AE}$ where $\rho_{AE}$ is given as input to the theorem (note that, here, the input state is independent of the subset chosen unlike in the ideal case). The process of choosing a subset $t$, measuring, and observing $q$ (resulting in post-measurement state $\rho(t,q)$) may be described, entirely, by the mixed state: \[\rho_{TQR}=\sum_tP_T(t)\kb{t}\otimes\sum_{q\in\alphabet_{d}^{|t|}}p(q|t)\kb{q}\otimes\rho(t,q),\]
where $p(q|t)$ is the probability of observing outcome $q$ given that the subset $t$ was sampled; here we use the ``R'' register to denote the remaining, unmeasured, portion of the state. Likewise, the ideal state, after performing this experiment, may be written as the mixed state: $\sigma_{TQR}=\sum_tP_T(t)\kb{t}\otimes\sum_{q}\tilde{p}(q|t)\kb{q}\otimes\sigma(t,q)$. We define $\Delta_{q,t}=\frac{1}{2}||\rho(t,q)-\sigma(t,q)||$, which may be treated as a random variable over the choice of $t$ and observed $q$. We want to show that, with high probability, $\Delta_{q,t}$ is ``small.''

It is not difficult to show that the expected value of $\Delta_{q,t}$ is ${\mathbb{E}}(\Delta_{q,t})=\mu \leq 2\epsilon$. Furthermore, the variance $V^2$ of this random variable has the property that $V^2 \leq \mu \leq 2\epsilon$ (see our proof in \cite{krawec2019quantum} for both these computations, though they follow immediately from properties of trace distance and the fact that $\Delta_{t,q} \le 1$).

Now, by Chebyshev's inequality, we have:
\begin{equation}
    \Pr(|\Delta_{q,t}-\mu| \geq \epsilon^\beta) \leq \frac{V^2}{\epsilon^{2\beta}} \leq 2\epsilon^{1-2\beta},
\end{equation}
(the last inequality follows since $\beta < \frac{1}{2}$); note that this probability is over all subsets $t$ and measurement outcomes $q$. Thus, except with probability at most $2\epsilon^{1-2\beta}$, after choosing $t$ and observing $q$, it holds that $|\Delta_{q,t}-\mu| \leq \epsilon^{\beta}$ which implies:
$$\frac{1}{2}||\rho(t,q)-\sigma(t,q)||=\Delta_{q,t} \leq \mu+\epsilon^{\beta} \leq 2\epsilon+\epsilon^{\beta}.$$
Thus, we may conclude that $\Hmin^{4\epsilon+2\epsilon^{\beta}}(A_Z|E)_\rho \geq \Hmin(A_Z|E)_\sigma$, completing the second step of the proof.

Of course, the above analysis assumed the input state $\rho_{AE}$ was pure. However, if the state is not pure, it may be purified and, incorporating this extra system to $E$, the result above follows.
\end{proof}

Notice that 
one may choose sampling strategies suitable to a particular application and, then, need only to analyze the classical strategy to attain a result in the quantum setting. 
Furthermore, arbitrary sampling strategies may be employed with arbitrary target functions, leading to a potential wide-range of applications. One simply needs to analyze the failure probabilities of the resulting classical sampling strategy (Equation \ref{eq:error-prob}). We demonstrate this by analyzing a QRNG protocol in the next section.


\subsection{Application to Quantum Random Number Generators}\label{section:QRNG}

Quantum Random Number Generators (QRNG) are protocols which, by utilizing a physical source of randomness in particular quantum sources, attempt to distill a uniform random string. For a cryptographic QRNG, the string should be uniform random and also independent of any adversary. At the most basic level, a QRNG protocol could consist of a source emitting a photon passing through a beam splitter connected to two photon counters. Such a system will lead to a random measurement on one detector or the other, producing a random stream of $0$'s and $1$'s. Such a setup assumes fully trusted devices (both the source and measurement apparatus are fully trusted and characterized and outside the control or influence of any adversary). 

On the opposite extreme is the fully device independent model \cite{di-qrng1,di-qrng2} whereby the source and measurement apparatus are not trusted (perhaps manufactured by the adversary - though one must still assume, of course, that the actual measurement outcome reported by the untrusted device cannot be sent to the adversary). Fully device independent protocols are obviously highly desirable from a cryptographic standpoint; however in practice, they are slow to implement \cite{di-qrng-exp,di-qrng-exp-2}. This leads to a middle-ground between these two extremes known as the \emph{source-independent} (SI) model introduced originally in \cite{vallone2014quantum} and studied further in several works including \cite{si-qrng2,si-qrng3}. Here, the quantum source is not trusted, however the measurement devices used are trusted and characterized. Such protocols are a step up from the fully trusted scenario (as they can take into account physical imperfections, but also the fact that an adversary may be entangled with the source and, thus, attempt to gain information on the resulting random string). Furthermore, they are highly practical, leading to Gbps implementations \cite{si-qrng-fast}. Finally, by not trusting the source, several fascinating possibilities are open, including the use of sunlight as the source \cite{si-qrng-sun}. For a general survey of QRNG protocols and their security models, the reader is referred to \cite{qrng-survey}.

In previous work, we showed that sampling-based entropic uncertainty relations provide optimistic results for QRNG protocols. In \cite{krawec2019quantum}, we analyzed a qubit-based protocol but without an adversary. In \cite{krawec2020new}, we analyzed a SI-QRNG protocol with an adversarial source and qudits ($d$-level systems), however where Alice was restricted to performing only a partial basis measurement (our previous relation could not take into account a full basis measurement for the sampling stage of the protocol). Here, we show how our entropic uncertainty relation can be used to provide highly optimistic bit generation rates for the full high-dimensional SI-QRNG protocol introduced in \cite{vallone2014quantum} (where a full basis measurement is required for the test stage). The protocol we analyze requires Alice to be able to measure in two bases $Z = \{\ket{0}, \cdots, \ket{d-1}\}$ and $X = \{\ket{x_0}, \cdots, \ket{x_{d-1}}\}$. We assume the measurement devices are fully characterized and so $\max_{i,j}|\braket{i|x_j}|$ is known. In the following we will assume that $|\braket{i|x_j}| = 1/\sqrt{d}$ for all $i,j$ however our analysis works identically for other scenarios. The protocol, then, operates as follows:
\begin{enumerate}
    \item \textbf{Preparation:} An adversary prepares a quantum state $\ket{\psi_0} \in \mathcal{H}_A\otimes\mathcal{H}_E$, where the $\mathcal{H}_A$ portion is an $(n+m)$-fold tensor of $\mathcal{H}_d$ (i.e., the $A$ register consists of $n+m$ qudits of dimension $d$ for a known $d \ge 2$). The $A$ portion is sent to Alice while the $E$ portion remains with the adversary. An ideal source should prepare the state $\ket{\psi_0} = \ket{x_0}^{\otimes (n+m)}\otimes\ket{\chi}_E$ - that is, a state independent of Eve and with $n+m$ perfect copies of the qudit state $\ket{x_0}$. As the source is adversarial, we do not assume anything about the structure of $\ket{\psi_0}$ other than it lives in $\mathcal{H}_A\otimes\mathcal{H}_E$.
    
    \item \textbf{Sampling and Measurements:} Alice chooses a random subset $t$ of size $m$ and measures those qudits indexed by $t$ in the $X$ basis, recording the outcome as $q \in \mathcal{A}_d^m$. The character counts of this will be used to determine how much information an adversary has (it should be that $c_0(q)$ is high). The remaining qudits she measures in the $Z$ basis, saving the resulting string as $r \in \mathcal{A}_d^n$. Note we are not considering experimental imperfections on the devices such as dark counts or low-efficiency detectors - we are only interested in the theoretical bound of ideal measurements, leaving these interesting practical measurement concerns as potential future work.
    
    \item \textbf{Post-Processing:} Alice runs a privacy amplification protocol, applying a two-universal hash function $f$ to the string $r$, resulting in her final random string $s = f(r)$. As proven in \cite{frauchiger2013true}, for a QRNG protocol of this nature, the hash function $f$ need only be chosen randomly once and then reused, so no additional randomness is needed here.
\end{enumerate}

The sampling portion of this protocol is easily seen to be $\samp_1$ introduced in Section \ref{section:sample-strats} with target function $g(x) = (c_0(x), \cdots, c_{d-1}(x))$. In this case, the size of the chosen subset $t$ is always $m$ leaving $n$ qudits unmeasured. So we write $J_q$ in place of $J_q^{(n)}$ from Theorem \ref{thm:twoparty} and its definition is:
\begin{equation}\label{eq:Jq-qrng}
J_q = \{i \in \alphabet_d^n \st \max_j |c_j(i) - c_j(q)| \le \delta\}.
\end{equation}
To apply the sampling based entropic uncertainty relation of Theorem \ref{thm:twoparty}, we first bound the size of this set. Of course $J_q \subset I_q = \{i \in \alphabet_d^n \st |w(i) - w(q)| \le \delta\}$ where $w(x)$ is the relative Hamming weight of $x$. Then, using the well-known volume of a Hamming ball, we may bound $|J_q| \le |I_q| \le d^{n\Hextd(w(q) + \delta)}$. This is the bound we used in our entropic uncertainty relation in \cite{krawec2020new} (which was based on the set $I_q$ not the full $J_q$ since full measurements were not supported in our earlier work). However, when we have full information on the string $q$, we may attempt to derive a tighter bound on $J_q$ itself for use in analyzing this QRNG protocol. Theorem \ref{Jq bound} provides an alternative bound on $|J_q|$ which is tighter in some scenarios as we discuss later.


\begin{theorem} \label{Jq bound}
Let $\frac{1}{d} > \delta > 0$ and $q\in\alphabet^m_d$ be given. Define the functions $\nu_i$ for each $i \in \al_d$, dependent on the choice of $q$, to be
\[\nu_i = \begin{cases}
0, & c_i(q) - \delta \leq 0 \\
c_i(q) - \delta, & \text{otherwise}.
\end{cases}\]
then, for $J_q = J_q^{(n)}$ defined in Equation \ref{eq:Jq-qrng}, we have:
\begin{equation}
    \log_2|J_q| \leq - n\sum_{i\in \al_d} \nu_i \log_2 \nu_i + n\log_2 n \left(1-\sum_{i\in \al_d} \nu_i\right) + (d+1)\log_2 e - \frac{d}{2}\log_2\left(\frac{1-d\delta}{d}\right).
\end{equation}
\end{theorem}

\begin{proof}
To prove this, we count the total number of ways one may construct a string with the required counts.
Let $\mathcal{K}_q = \left\{(x_0, \dots, x_{d-1}) \in \mathbb{N}^d : |x_i - nc_i(q)| \leq n\delta \text{ and } \sum x_i = n\right\}$ and observe that

\begin{align*}
|J_q| &= \sum_{k \in \mathcal{K}_q} \prod_{k_i \in k} {n - \sum_{j=0}^{i-1} k_j \choose k_i} \\
&= \sum_{k \in \mathcal{K}_q} \frac{n!}{k_0!(n-k_0)!}\cdot\frac{(n-k_0)!}{k_1!(n-k_0-k_1)!}\cdot\frac{(n-k_0-k_1)!}{k_2!(n-k_0-k_1-k_2)!}\dots \\
&= \sum_{k \in \mathcal{K}_q} \frac{n!}{k_0!k_1!k_2!\dots} \;\;=\;\; n!\sum_{k \in \mathcal{K}_q}\prod_{k_i \in k} \frac{1}{k_i!}.
\end{align*}

Let $\mathcal{M}_q = \{(x_0,\dots,x_{d-1})\in\mathbb{N}^d : |x_i - nc_i(q)| \leq n\delta\}$. Of course $\mathcal{K}_q \subset \mathcal{M}_q$. This immediately implies
\[n!\sum_{k \in \mathcal{K}_q}\prod_{k_i \in k} \frac{1}{k_i!} \leq n!\sum_{x \in \mathcal{M}_q}\prod_{x_i \in x} \frac{1}{x_i!}.\]
Now let $\{x_i^1,x_i^2,\dots,x_i^{m_i}\} \subset \N$ be the values in increasing order which satisfy $|x_i^j-nc_i(q)| \leq n\delta$ for all $j \in \{1,\dots,m_i\}$. We can enumerate the set $\mathcal{M}_q$ as
\[\mathcal{M}_q = \{(x_0^{j_0},x_1^{j_1},\dots,x_{d-1}^{j_{d-1}}) \st j_i \in \{1,\dots,m_i\} \;\forall i \in \{0,\dots,d-1\}\}.\]
Then
\begin{align*}
    \sum_{x \in \mathcal{M}_q} \prod_{i=0}^{d-1} \frac{1}{x_i!} &= \sum_{j_0,\dots,j_{d-1}} \left(\frac{1}{x_0^{j_0}!}\cdot \frac{1}{x_1^{j_1}!} \cdot \ldots \cdot \frac{1}{x_{d-1}^{j_{d-1}}!}\right) \\
    &= \prod_{i=0}^{d-1}\left(\frac{1}{x_i^{1}!} + \frac{1}{x_i^{2}!} + \ldots + \frac{1}{x_i^{m_i}!}\right) \;=\; \prod_{i=0}^{d-1}\sum_{j_i=1}^{m_i} \frac{1}{x_i^{j_i}!}
\end{align*}

The benefit of isolating these partial sums of $1/x_i^{j_i}!$ is that we can take advantage of the Taylor series for $e^x$ to bound this partial sum. We can expand on this to get the following:

\begin{align*}
    n!\prod_{i=0}^{d-1}\left(\sum_{j_i=1}^{m_i} \frac{1}{x_i^{j_i}!}\right) \;&=\; n!\prod_{i=0}^{d-1} \frac{1}{x_i^1!} \left(\sum_{j_i=1}^{m_i} \frac{1}{(x_i^{j_i}!)/(x_i^1!)}\right) \;\leq\; n!\prod_{i=0}^{d-1} \frac{1}{x_i^1!} \left(\sum_{j_i=1}^{m_i} \frac{1}{j_i!}\right) \\
    &\leq\; n!\prod_{i=0}^{d-1} \frac{e}{x_i^1!} \;=\; n!\cdot e^d \cdot \prod_{i=0}^{d-1}\frac{1}{x_i^1!}.
\end{align*}

Since each $x_i \ge 0$, we replace the value of $x_i^1$ with the value $n\nu_i$ for each $i$, where $\nu_i$ is defined in the Theorem statement. Furthermore, below, since $0! = 1$, we only need to multiply by those $\nu_i > 0$. Then,

\allowdisplaybreaks
\begin{align}
    \log_2 |J_q| &\leq \log_2 \left(n!\cdot e^d \cdot \prod_{i=0}^{d-1}\frac{1}{x_i^1!}\right) \notag \\
    &= \log_2 \left(n!\cdot e^d \cdot \prod_{\nu_i \neq 0} \frac{1}{\lceil n\nu_i \rceil !}\right) \notag \\
    &= \log_2(n!) + d\log_2 e - \sum_{\nu_i\neq 0} \log_2(\lceil n\nu_i \rceil !) \notag \\
    &\leq \log_2\left(en^{n+1/2}e^{-n}\right) + d\log_2 e - \sum_{\nu_i\neq 0} \log_2 \left(\sqrt{2\pi}(n\nu_i)^{n\nu_i+1/2}e^{-n\nu_i}\right) \label{eq:Stirling} \\
    &\leq n\log n + (d+1-n)\log_2 e + \frac{1}{2}\log_2 n - \sum_{\nu_i\neq 0} \left( (n\nu_i+1/2)\log_2 n\nu_i -n\nu_i \log_2 e\right) \notag \\
    &\leq - n\sum_{\nu_i\neq 0} \nu_i \log_2 \nu_i + n\log_2 n \left(1-\sum_{\nu_i\neq 0} \nu_i\right) + (d+1)\log_2 e \notag \\
    &\;\;+ \frac{1}{2}\left(\log_2 n - \sum_{\nu_i\neq 0} \log_2 n\nu_i\right) \label{eq:condense-d+1} \\
    &\leq - n\sum_{i\in \al_d} \nu_i \log_2 \nu_i + n\log_2 n \left(1-\sum_{i \in \al_d} \nu_i\right) + (d+1)\log_2 e \notag \\
    &\;\;+ \frac{1}{2}\left(\log_2 n - d\log_2\left(\frac{n(1-d\delta)}{d}\right)\right) \label{eq:concavity-monotonicity} \\
    &\leq - n\sum_{i\in \al_d} \nu_i \log_2 \nu_i + n\log_2 n \left(1-\sum_{i\in \al_d} \nu_i\right) + (d+1)\log_2 e - \frac{d}{2}\log_2\left(\frac{1-d\delta}{d}\right). \notag
\end{align}
Inequality \ref{eq:Stirling} follows from the Stirling upper and lower bounds. Then, the $(d+1)$ in inequality \ref{eq:condense-d+1} follows from $-n(1-\sum_i \nu_i)\log_2 e \leq 0$. Jensen's Inequality and concavity of the logarithm imply inequality \ref{eq:concavity-monotonicity}.
\end{proof}

Now we use Theorems \ref{thm:twoparty} and \ref{Jq bound} to analyze the protocol described above. Let $\epsilon > 0$ be arbitrarily chosen by the user (this will determine the user's desired failure probability and security properties).  We use
\begin{equation}\label{eq:QRNG-delta}
\delta = \sqrt{\frac{(m+n+2)\ln(2d/\epsilon^2)}{m(m+n)}},
\end{equation}
which, by Lemma \ref{lemma:samp-psi1} implies that the failure probability will be $\epsilon^2$ (and so the $\epsilon$ in Theorem \ref{thm:twoparty} will match the chosen value of $\epsilon$ here). Finally, let
$\epsilon_{PA} = 4\epsilon^\beta + 9\epsilon$ be the distance from an ideal uniform random string of size $\ell$ independent of $E$'s system. 

Using Theorem \ref{thm:twoparty} along with privacy amplification (Equation \ref{eq:PA}), we have that, except with probability at most $2\epsilon^{1-2\beta}$, the number of uniform random bits extracted from the protocol leading to an $\epsilon_{PA}$ secure string is:
\begin{equation}
    \ell_{\text{ours}} = n\log_2 d - \log_2|J_q| - 2\log_2 \frac{1}{\epsilon},
\end{equation}
where
\begin{equation}\label{eq:min-Jq}
    \log_2|J_q| \leq \min\left\{\mathcal{F}, \mathcal{G} \right\},
\end{equation}
\begin{align}
    \mathcal{F} &= - n\sum_{i\in \al_d} \nu_i \log_2 \nu_i + n\log_2 n \left(1-\sum_{i\in \al_d} \nu_i\right) + (d+1)\log_2 e - \frac{d}{2}\log_2\left(\frac{1-d\delta}{d}\right), \\
    \mathcal{G} &= n\Hextd_d(1-\nu_0)\log_2 d\label{eq:Jq-G}
\end{align}
by Theorem \ref{Jq bound} and the standard bound on the volume of a Hamming ball as discussed earlier. In our evaluations, we set $\epsilon = 10^{-36}$ and $\beta = 1/3$ which balances the failure probability of Theorem \ref{thm:twoparty} (namely, the probability of failure is $2\epsilon^{1-2\beta}$) and the smoothing parameter used in the min entropy. With these settings, the failure probability and the value of $\epsilon_{PA}$ are on the order of $10^{-12}$.

We compare our new lower bound $\ell_{\text{ours}}$ for this protocol against the lower bound provided in \cite{vallone2014quantum} using alternative methods and an alternative entropic uncertainty relation. We also compare with another high-dimensional SI-QRNG from \cite{xu2016experimental}. Note that, due to our bound on $J_q$ in Equation \ref{eq:min-Jq}, our new result here will never be worse then the SI-QRNG protocol analyzed in our prior work \cite{krawec2020new} (which used Equation \ref{eq:Jq-G} only) and so we do not compare with that here.

A lower bound for the SI-QRNG protocol of \cite{vallone2014quantum}, which we denote here as $\ell_1$, was given in that reference by:
\[\ell_1 \geq n\left(\log_2 d - 2\log_2\left[\frac{\Gamma(m+d)}{\Gamma(m+d+\frac{1}{2})} \sum_{i=0}^{d-1} \frac{\Gamma(c_i+\frac{3}{2})}{\Gamma(c_i+1)} \right]\right),\]
where $m$ is the test size and $c_i$ represents the number of measurement outcomes that result in outcome $\ket{x_i}$.

The protocol of \cite{xu2016experimental} is slightly different from the one we analyze. Here, an adversarial source prepares an entangled high-dimensional state (if the source were honest, it would prepare $n+m$ copies of the state $\ket{\psi_0} = \frac{1}{\sqrt{d}}\sum_{i=0}^{d-1}\ket{i,i}_{A_1A_2}$) sending the $A_1$ and $A_2$ registers to Alice. Alice chooses a random subset and measures the $A_1$ and $A_2$ qudit systems each in a $d$-dimensional basis $X$ resulting in classical characters $c_{A_1}(i)$ and $c_{A_2}(i)$ corresponding to the $i$th iteration of registers $A_1$ and $A_2$. For the remaining unmeasured systems, she discards the $A_2$ system and measures only the $A_1$ system in the $Z$ basis resulting in her secret string. If the source were honest, it should be that the $X$ basis measurement outcomes of the $A_1$ and $A_2$ register are fully correlated. She then applies privacy amplification to the result of the $Z$ basis measurement. A lower bound for the number of random bits that may be extracted from this protocol, which we denote here as $\ell_2$, was computed in \cite{xu2016experimental}:
\[\ell_2 = n\log_2 d - \log_2 \gamma(d_0 + \delta'),\]
where
\begin{equation}\label{eq:gamma-fct}
    \gamma(x) = (x+\sqrt{1+x^2})\left(\frac{x}{\sqrt{1+x^2}-1}\right)^x
\end{equation}
and
\begin{equation}\label{eq:delta-prime}
    \delta' = d\sqrt{\frac{N^2}{n^2m}\ln\left(\frac{4}{\epsilon}\right)}.
\end{equation}
The term $d_0$ is computed as the average difference between the measurements values of the pairs, $c_{A_1}(i)$ and $c_{A_2}(i)$ for $i$ from 1 to $m$. That is,
$d_0 = \frac{1}{m}\sum_{i=1}^m |c_{A_1}(i)-c_{A_2}(i)|.$

For all protocols, we assume a failure probability on the order of $10^{-12}$. The difference from the ideal random string (Equation \ref{eq:PA}) is also set to be $10^{-12}$. As we are only interested in comparing the relative performance, we do not consider the additional randomness used to choose a random subset of size $m$. Since all protocols in our evaluation are using the same process for this and same sampling sizes (in particular we use $7\%$ of all signals for sampling), they will each lose the same amount from their respective $\ell$ values and so the comparison remains unchanged.

Note that for each of the three bounds, no assumption is needed on the noise in the channel - Alice simply uses the direct measurement result from the test case (in the $X$ basis) and evaluates $\ell$. To compare, however, we will simulate certain noise scenarios. We first compare these protocols assuming a depolarization channel acting on each qudit state independently and identically. Such a channel will cause the qudit to become the completely mixed state with some probability $Q$; otherwise it remains in its original state. In this setting, we see the protocol of \cite{vallone2014quantum}, \emph{but augmented using our new entropic uncertainty relation here} outperforms both $\ell_1$ and $\ell_2$. Since $\ell_1$ is the same protocol we are analyzing with $\ell_{ours}$ this shows the great benefit of sampling-based entropic uncertainty relations. This evaluation is shown in Figures \ref{fig:QRNG_lo-dim_dep}, \ref{fig:QRNG_hi-dim_dep}. Next, we evaluate on asymmetric channels which are more likely to add noise towards one basis vector over another (i.e., it is more likely to change a $\ket{0}$ to a $\ket{1}$ as opposed to changing a $\ket{0}$ to a $\ket{3}$). Depending on the state favored by the channel, our bound generally outperforms prior work as shown in Figures \ref{fig:QRNG_lo-dim_nondep2} (right),  \ref{fig:QRNG_lo-dim_nondep1} and \ref{fig:QRNG_hi-dim_nondep} (right), though there are scenarios where the protocol of \cite{xu2016experimental} can outperform our analysis as shown in Figures \ref{fig:QRNG_lo-dim_nondep2} (left) and \ref{fig:QRNG_hi-dim_nondep} (left). Note, however, that the protocol of \cite{xu2016experimental} is a different protocol; our methods applied to that protocol may provide a boost in performance in this scenario also, a question we leave as future work. Comparing $\ell_{ours}$ with $\ell_1$, which is the generation rate for the same protocol of \cite{vallone2014quantum}, shows that our new entropic uncertainty relation always leads to more optimistic bit generation rates in every scenario we simulated. Also, note that in all cases (including in the case highlighted in Figures \ref{fig:QRNG_lo-dim_nondep2} and \ref{fig:QRNG_hi-dim_nondep}), if we take the number of signals to be high enough, our bound outperforms.

\begin{figure}
    \centering
    \includegraphics[width=.46\textwidth]{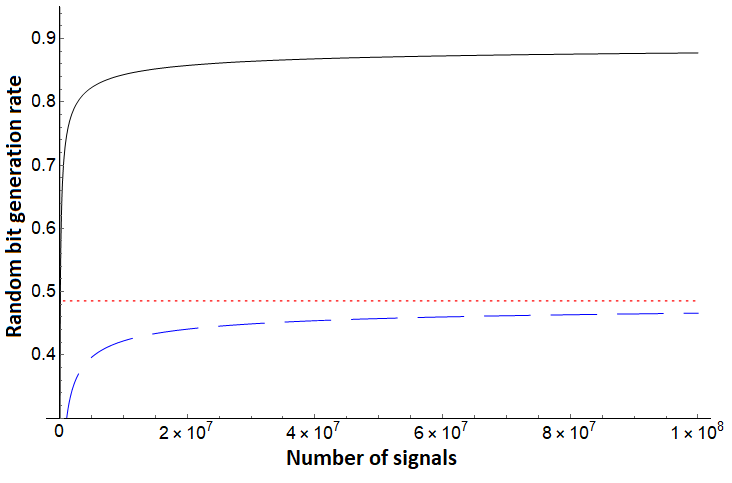}
    \includegraphics[width=.46\textwidth]{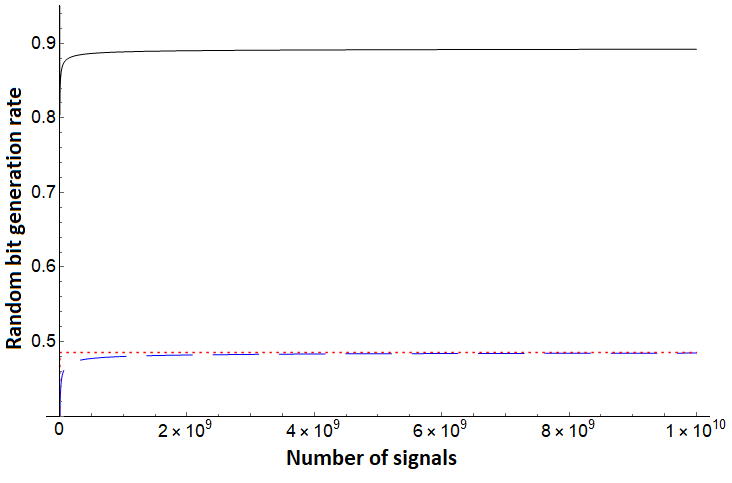}
    \caption{{\footnotesize Random bit generation rates. $x$-axis: Total number of signals $N$ from which $.07N$ are used for sampling; $y$-axis: Random bit generation rate $\ell/N$. Solid black: $\ell_{\text{ours}}/N$; Dotted red: $\ell_1/N$ from~\cite{vallone2014quantum} (same protocol, different security analysis); Dashed blue: $\ell_2/N$ from~\cite{xu2016experimental} (different protocol and different security analysis method). Both graphs plot $d=4$ with $c(q) = (0.8,1/15,1/15,1/15)$ (recall that $c(q)$ denotes the $d$-tuple of character counts as discussed in Section \ref{section:notation}). The left and right graphs plot $N \leq 10^8$ and $N \leq 10^{10}$ respectively.}}
    \label{fig:QRNG_lo-dim_dep}
\end{figure}

\begin{figure}
    \centering
    \includegraphics[width=.48\textwidth]{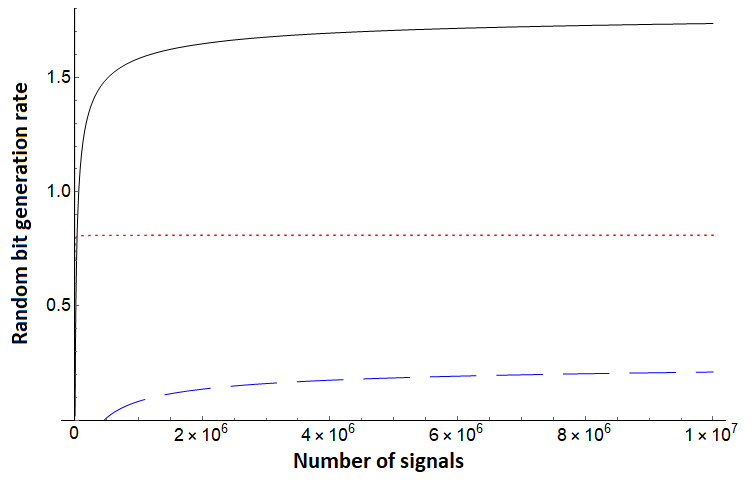}
    \includegraphics[width=.48\textwidth]{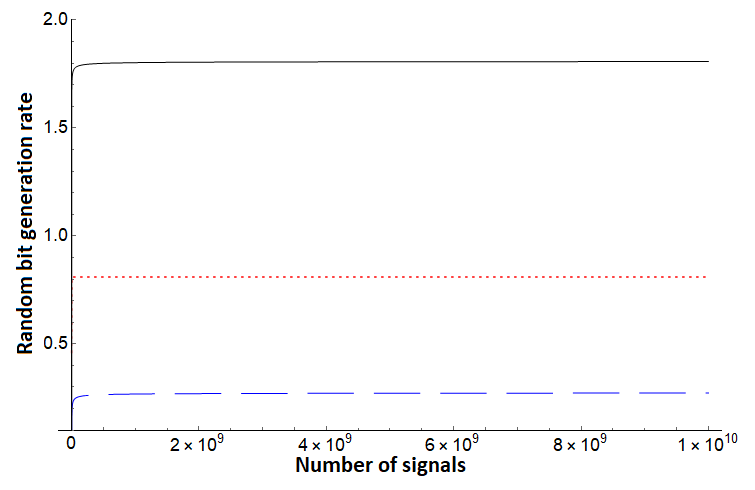}
    \caption{{\footnotesize Random bit generation rates. $x$-axis: Total number of signals $N$; $y$-axis: Random bit generation rate $\ell/N$. Solid black: $\ell_{\text{ours}}/N$; Dotted red: $\ell_1/N$ from~\cite{vallone2014quantum}; Dashed blue: $\ell_2/N$ from~\cite{xu2016experimental}. Both graphs plot $d=16$ with $c(q) = (0.7,0.02,0.02,\dots,0.02)$. The left and right graphs plot $N \leq 10^8$ and $N \leq 10^{10}$ respectively.}}
    \label{fig:QRNG_hi-dim_dep}
\end{figure}

\begin{figure}
    \centering
    \includegraphics[width=.48\textwidth]{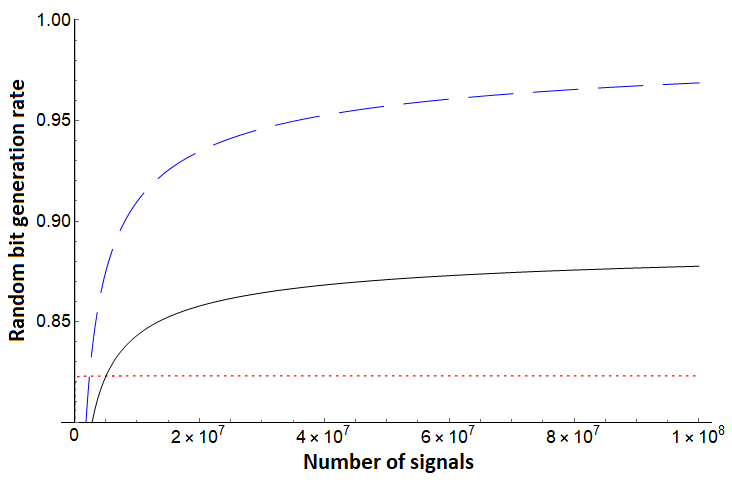}
    \includegraphics[width=.48\textwidth]{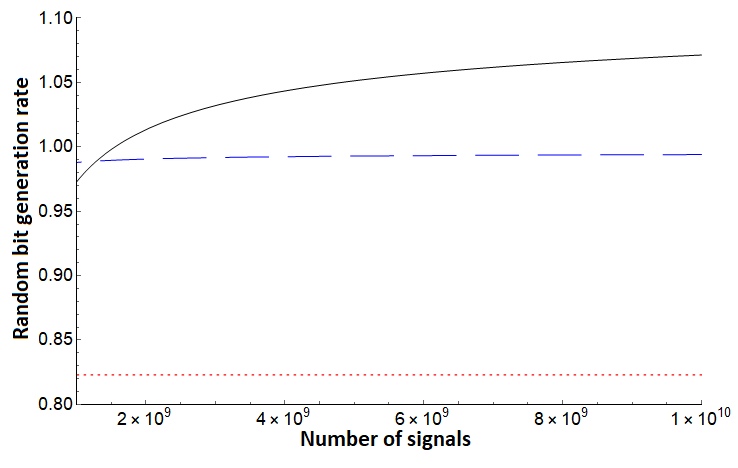}
    \caption{{\footnotesize Random bit generation rates. $x$-axis: Total number of signals $N$; $y$-axis: Random bit generation rate $\ell/N$. Solid black: $\ell_{\text{ours}}/N$; Dotted red: $\ell_1/N$ from~\cite{vallone2014quantum}; Dashed blue: $\ell_2/N$ from~\cite{xu2016experimental}. Both graphs plot $d=4$ with $c(q) = (0.8,0.19,0.005,0.005)$. The left and right graphs plot $N \leq 10^8$ and $10^9 \leq N \leq 10^{10}$ respectively.}}
    \label{fig:QRNG_lo-dim_nondep2}
\end{figure}

\begin{figure}
    \centering
    \includegraphics[width=.48\textwidth]{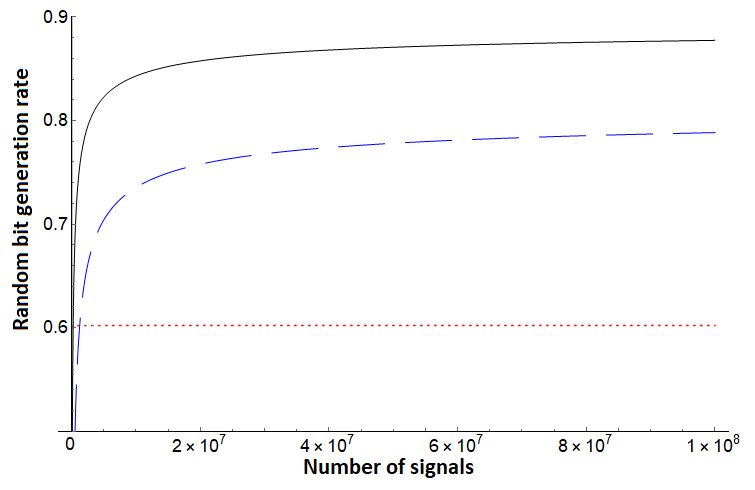}
    \includegraphics[width=.48\textwidth]{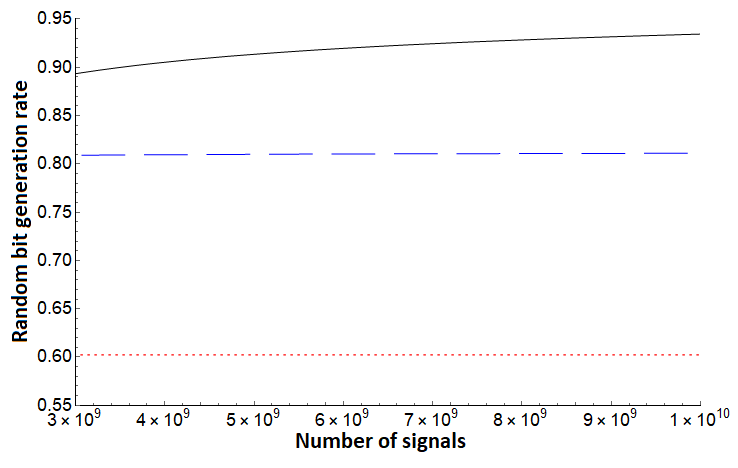}
    \caption{{\footnotesize Random bit generation rates. $x$-axis: Total number of signals $N$; $y$-axis: Random bit generation rate $\ell/N$. Solid black: $\ell_{\text{ours}}/N$; Dotted red: $\ell_1/N$ from~\cite{vallone2014quantum}; Dashed blue: $\ell_2/N$ from~\cite{xu2016experimental}. Both graphs plot $d=4$ with $c(q) = (0.8,0.15,0.025,0.025)$. The left and right graphs plot $N \leq 10^8$ and $3\times 10^9 \leq N \leq 10^{10}$ respectively.}}
    \label{fig:QRNG_lo-dim_nondep1}
\end{figure}


\begin{figure}
    \centering
    \includegraphics[width=.48\textwidth]{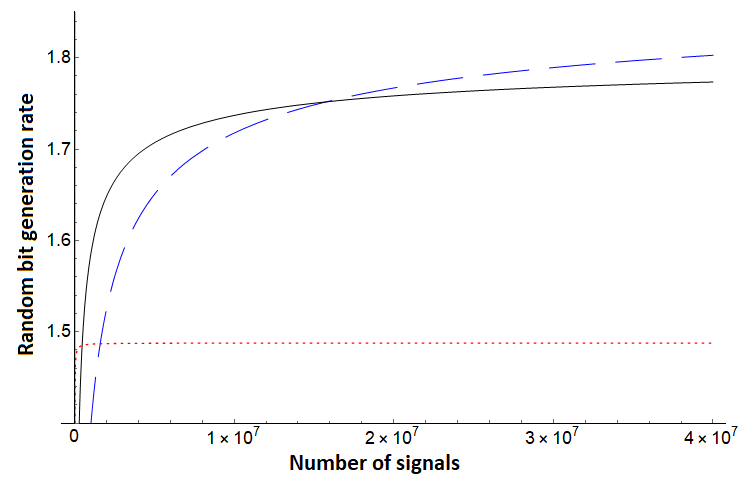}
    \includegraphics[width=.48\textwidth]{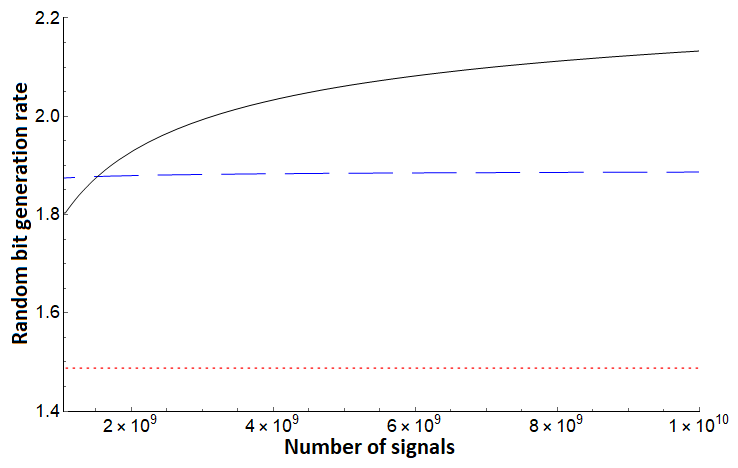}
    \caption{{\footnotesize Random bit generation rates. $x$-axis: Total number of signals $N$; $y$-axis: Random bit generation rate $\ell/N$. Solid black: $\ell_{\text{ours}}/N$; Dotted red: $\ell_1/N$ from~\cite{vallone2014quantum}; Dashed blue: $\ell_2/N$ from~\cite{xu2016experimental}. Both graphs plot $d=16$ with $c(q) = (0.7,0.16,0.075,0.035,0.0025,0.0025,\dots,0.0025)$. The left and right graphs plot $N \leq 4\times 10^7$ and $1.05 \times 10^9 \leq N \leq 10^{10}$ respectively.}}
    \label{fig:QRNG_hi-dim_nondep}
\end{figure}


In summary, Figures \ref{fig:QRNG_lo-dim_dep} and \ref{fig:QRNG_hi-dim_dep} highlight how $\ell_{\text{ours}}$ consistently outperforms $\ell_1$~\cite{vallone2014quantum} and $\ell_2$~\cite{xu2016experimental} on a depolarization channel for different dimensions $d$. Our bound for $\ell_{\text{ours}}$ still performs very well on systems far from depolarization, as shown in Figure \ref{fig:QRNG_lo-dim_nondep1}. However, there can exist quantum channels which lead to $\ell_{\text{ours}}$ producing a lower random bit generation rate than $\ell_2$ for certain $N$. Even in these cases, Figures \ref{fig:QRNG_lo-dim_nondep2} and \ref{fig:QRNG_hi-dim_nondep} highlight that, assuming sufficient computational power to process larger blocks in the post-processing stage of the protocol, $\ell_{\text{ours}}$ can produce a much higher random bit generation rate than $\ell_1$ and $\ell_2$ on a large block of signals.


\subsection{Asymptotic Behavior and Analysis}


In Equation \ref{eq:min-Jq}, we take the bound for $\log_2|J_q|$ to be the minimum of Theorem \ref{Jq bound} and the size of a Hamming ball. The reason for doing so is that while the bound on the size of a Hamming ball is tighter for some scenarios, the bound from Theorem \ref{Jq bound} is significantly better in others, especially, as our numerical simulations show, for large numbers of signals. In this section, we analyze and compare the asymptotic behavior of both bounds.  We will also use this work to show an alternative proof of the famous Maassen-Uffink relation \cite{maassen1988generalized} for high dimensional systems.

First, we prove a technical lemma about the relation between the $d$-ary entropy function $H_d$ and the Shannon entropy, which will be needed to analyze the asymptotic behavior.


\begin{lemma} \label{d-ary and Shannon relation}
Let $X$ be a discrete random variable with $d$ possible outcomes $\{x_0, \dots, x_{d-1}\}$ such that the probability of observing outcome $x_j$ is $p_j$ for each $j$. Then for any $i$, it holds that
\[H_d(1-p_i) \geq \log_d 2 \cdot H(X)\]
where equality holds if and only if $p_j = p_k$ for all $j,k \neq i$ (i.e., if the distribution is uniform on the other outcomes not equal to $i$).
\end{lemma}

\begin{proof}
Fix $i \in \{0,\dots,d-1\}$ and let $Y$ be the random variable where the probability of observing $x_i$ is $q_i=p_i$ and the probability of observing $x_j$ for any $j \neq i$ is $q_j = \frac{1-p_i}{d-1}$. Then
\begin{align*}
    H_d(1-p_i) &= (1-p_i)\log_d(d-1) - (1-p_i)\log_d(1-p_i) - p_i\log_d(p_i) \\
    &= \log_d 2 \left[-p_i\log_2(p_i) - \sum_{j \neq i} \frac{1-p_i}{d-1} \log_2\left(\frac{1-p_i}{d-1}\right)\right] \\
    &= \log_d 2 \left[-p_i\log_2(p_i) - \sum_{j \neq i} q_j \log_2\left(q_j\right)\right] \;=\; H(Y) \cdot \log_d 2.
\end{align*}
Moreover, observe that
\begin{align*}
    H(X) \;&=\; -p_i\log_2 (p_i) - \sum_{j \neq i}p_j \log_2(p_j) \\
    &\leq\; -p_i\log_2 (p_i) - \sum_{j \neq i}q_j \log_2(q_j) \;=\; H(Y).
\end{align*}
Note the inequality is shown by recalling that Shannon entropy is maximal if and only if given a uniform distribution (which also proves equality if the distribution is uniform on outcomes other than $x_i$).
\end{proof}

We now show that our bound for $\log_2|J_q|/n$ converges to the Shannon entropy of the random variable induced by a measurement on some i.i.d. system. This will then lead us to an alternative proof of the Maassen-Uffink relation from \cite{maassen1988generalized}.

\begin{lemma} \label{Shannon convergence}
Let $\rho$ be a quantum state acting on $\mathcal{H}_d$ and consider the $n$-fold tensor state $\rho' = \rho^{\otimes n}$. Furthermore, let $\delta = O\left(\frac{1}{\sqrt{n}}\right)$. Consider measuring all $n$ qudits of the state $\rho'$ in some $d$-dimensional orthonormal basis $X$ resulting in some $q\in\al_d^n$ and from this define the set $J_q = \{i \in \al_d^n : \max_j|c_j(i) - c_j(q)| \leq \delta\}$ as before. Then it follows that
\[\lim_{n\to\infty} \frac{\log_2 |J_q|}{n} \leq H(X)_\rho.\]
\end{lemma}

\begin{proof} Define $\nu_i = \max\{c_i(q) - \delta,0\}$ and let
\[\mathcal{F}(q, n, d, \delta) = - n\sum_{i\in \al_d} \nu_i \log_2 \nu_i + n\log_2 n \left(1-\sum_{i\in \al_d} \nu_i\right) + (d+1)\log_2 e - \frac{d}{2}\log_2\left(\frac{1-d\delta}{d}\right).\]
Observe that $\frac{\mathcal{F}(q,n,d,\delta)}{n} \leq -\sum_{i} \nu_i \log_2(\nu_i) + d\delta\log_2 n + \frac{O(1)}{n}$ where we used the fact that $\sum_i \nu_i \geq 1-d\delta$. Since $\delta = O\left(\frac{1}{\sqrt{n}}\right)$ we have
\[\frac{\mathcal{F}(q,n,d,\delta)}{n} = -\sum_{i\in \al_d}\nu_i \log_2(\nu_i) + O\left(\frac{\log_2 n}{\sqrt{n}}\right) + O\left(\frac{1}{n}\right).\]

\noindent Then, $\nu_i = \max\{c_i(q) - \delta,0\} \to p_i$ as $n \to \infty$ by the law of large numbers and the assumption on $\delta$, where we use $p_i$ to denote the probability of observing $\ket{x_i}$, the $i$'th basis vector in the measurement basis $X$. Hence,
\[\sum_{i \in \al_d} \nu_i \log_2 (\nu_i) \to \sum_{i \in \al_d} p_i \log_2 (p_i).\]
Finally, by Theorem \ref{Jq bound}, we have $\log_2|J_q| \leq \mathcal{F}(q, n, d, \delta)$, and so we conclude
\begin{align*}
    \lim_{n\to\infty} \frac{\log_2|J_q|}{n}\leq \lim_{n \to \infty} \frac{\mathcal{F}(q,n,d,\delta)}{n} &\leq \lim_{n\to \infty} \left(-\sum_{\nu_i \neq 0}\nu_i \log_2(\nu_i) + O\left(\frac{\log_2n}{\sqrt{n}}\right) + O\left(\frac{1}{n}\right)\right) \\
    &= -\sum_{i \in \al_d} p_i \log_2(p_i) \;\;=\;\; H(X)_\rho
\end{align*}
\end{proof}

Now we are ready to show that our bound for $\log_2 |J_q|$ grows at most as quickly as the volume of a Hamming ball (used in our earlier work in \cite{krawec2020new}). How much slower our bound grows asymptotically depends on the observed relative counts from the sampled $q$.

\begin{theorem} \label{log old bound vs log new bound}
Let $\nu_i = \max\{c_i(q) - \delta,0\}$ for any $i \in \al_d$. Let $\mathcal{G}(q, n, d, \delta) = \frac{n\Hextd_d(1-\nu_a)}{\log_d 2}$, where $a$ is the element of $\al_d$ such that $c_a(q) = \max_{i \in \al_d} c_i(q)$, and
\[\mathcal{F}(q, n, d, \delta) = - n\sum_{i\in \al_d} \nu_i \log_2 \nu_i + n\log_2 n \left(1-\sum_{i\in \al_d} \nu_i\right) + (d+1)\log_2 e - \frac{d}{2}\log_2\left(\frac{1-d\delta}{d}\right).\]
 Let $\delta$ depend on $n$ and $\delta = O\left(\frac{1}{\sqrt{n}}\right)$ asymptotically. Then for arbitrary quantum state $\rho$ acting on $\mathcal{H}_d$, we have that
\[\lim_{n \to \infty} \frac{\mathcal{F}(q, n, d, \delta)}{\mathcal{G}(q, n, d, \delta)} \leq 1\]
where equality holds if and only if $p_i = \frac{1-p_a}{d-1}$ for all $i \neq a$ given $p_k$ is the probability of observing outcome $k$ in basis $X$ (measuring $\rho$).
\end{theorem}

\begin{proof} Notice that by the proof of Lemma \ref{Shannon convergence}
\begin{align*}
    \lim_{n\to\infty} \frac{\mathcal{F}(q, n, d, \delta)}{\mathcal{G}(q, n, d, \delta)} &= \lim_{n\to\infty} \frac{\log_d 2}{\Hextd_d(1-\nu_a)} \cdot \lim_{n\to\infty} \frac{\mathcal{F}(q, n, d, \delta)}{n} \\
    &\leq H(X)_\rho \cdot \lim_{n\to\infty} \frac{\log_d 2}{\Hextd_d(1-\nu_a)}.
\end{align*}
Then, by Lemma \ref{d-ary and Shannon relation} and the definition of $\Hextd_d$, it follows that
\[\Hextd_d(1-\nu_a) \geq H_d(1-\nu_a) \geq \log_d 2 \cdot H(X)_\rho\]
and hence
\[H(X)_\rho \cdot \lim_{n\to\infty} \frac{\log_d 2}{\Hextd_d(1-\nu_a)} = \lim_{n\to\infty} \frac{\log_d 2 \cdot H(X)_\rho}{\Hextd_d(1-\nu_a)} \leq 1.\]
where equality holds if and only if $\log_d 2\cdot H(X) = \lim_{n\to\infty} \Hextd_d(1-\nu_a)$. This is true if and only if $p_i = \frac{1-p_a}{d-1}$ by Lemma \ref{d-ary and Shannon relation} and the fact that $\nu_a \to p_a$ as $n \to \infty$.
\end{proof}


\subsubsection{Alternative Proof of Maassen-Uffink Relation}

With the above analysis, our Theorems \ref{thm:twoparty} and \ref{Jq bound} can be used to provide an alternative proof of the Maassen and Uffink entropic uncertainty relation for projective basis measurements of $d$-dimensional states. Note that in \cite{krawec2019quantum} we showed quantum sampling can be used to provide an alternative proof of this relation but only for the qubit case. Furthermore, our earlier work in \cite{krawec2020new} also cannot lead to an alternative proof of this relation in the high dimensional ($d\ge 3$) case as Theorem \ref{log old bound vs log new bound} shows.

\begin{corollary} \label{Maassen and Uffink}
Let $Z = \{\ket{z_i}\}_{i \in \al_d}$ and $X = \{\ket{x_i}\}_{i \in \al_d}$ be two orthonormal bases and let $\rho$ be a density operator acting on $\mathcal{H}_d$. Then, except with arbitrarily small probability, it holds that
\[H(Z)_\rho + H(X)_\rho \geq \gamma,\]
where $\gamma = -\log\max_{i,j}|\braket{z_i|x_j}|^2$.
\end{corollary}

\begin{proof}
Consider the state $\rho' = \rho^{\otimes 2n}$. We apply Theorem \ref{thm:twoparty} to $\rho'$ using sampling strategy $\samp_1$ with $m = n$. Since $\rho'$ is i.i.d., for any subset $t$ of size $n$ and any measurement outcome $q$ on that subset, the post-measurement state is simply $\rho^{\otimes n}$.

Fix $\widehat{\epsilon} > 0$ and $0 < \beta < 1/2$. Then, for any $n$ and $\epsilon \leq \widehat{\epsilon}$, setting $\delta = \sqrt{\frac{(n+1)\ln(2d/\epsilon^2)}{n^2}}$, Theorem \ref{thm:twoparty} implies that, except with probability at most $\widehat{\epsilon}^{1-2\beta}$, the inequality
\[\frac{1}{n}\Hmin^{4\epsilon + 2\epsilon^\beta}(Z|E)_{\rho(t,q)} + \frac{1}{n}\log_2 |J_q^{(n)}| \geq \gamma\]
holds, where $q$ is the observed value after measuring using $Z$. Then, by the asymptotic equipartition property, it follows that
\[
\lim_{\epsilon \to 0} \lim_{n \to \infty} \frac{1}{n}\Hmin^{4\epsilon + 2\epsilon^\beta}(Z|E)_{\rho^{\otimes n}} = H(Z|E)_\rho.
\]
This, combined with Lemma \ref{Shannon convergence} and the fact that $H(Z) \ge H(Z|E)$, completes the proof.


\end{proof}



\section{A Three-Party Sampling-Based Entropic Uncertainty Relation}

We now turn our attention to deriving a new three-party sampling-based entropic uncertainty relation involving Alice, Bob, and Eve. Later we show an application to a finite key analysis of the high-dimensional BB84 \cite{HD-BB84}. To begin, consider the following experiment, extending an earlier version to this three party case: on an input state of the form $\rho_{TABE} = \sum_{t_A,t_B}p(t_A,t_B)\kb{t_A,t_B}\otimes\rho_{ABE}^{t_A,t_B}$, choose a random subset $t = (t_A,t_B)$ by measuring the $T$ register, causing the state to collapse to $\rho_{ABE}^{t_A,t_B}$ (though, as before, this $\rho$ portion may be independent of the chosen subset in which case a random subset is chosen which does not affect the rest of the input state). We assume $|t_A| = |t_B| = m$. Next, a portion of the $A$ and $B$ registers, indexed by the chosen subsets, are measured in basis $X = \{\ket{x_0}, \cdots, \ket{x_{d-1}}\}$ resulting in outcome $q_A, q_B \in \alphabet_d^m$. This measurement causes the remaining state to collapse to $\rho_{ABE}(t, q_A,q_B)$. The experiment outputs $(t, q_A,q_B, \rho_{ABE}(t,q_A,q_B))\experiment{\rho_{TABE}, X}$.

Note that technically, by considering Alice and Bob as one party for the sampling portion, one could potentially use Theorem \ref{thm:twoparty} with a suitable sampling strategy similar to $\samp_0$ or $\samp_1$. However, this would bound the resulting min entropy as a function of the set $J_{q_A,q_B} = \{(i,j)\in\mathcal{A}_d^{2n} \st |\hd(q_A,q_B) - \hd(i,j)| \le \delta\}$. It is not difficult to see that $|J_{q_A,q_B}| \ge d^n$ (since for any fixed $q_A$ and $q_B$, and for every $i \in \mathcal{A}_d^n$, one may find a $j\in\mathcal{A}_d^n$ satisfying $(i,j) \in J_{q_A,q_B}$). Recalling from our Theorem that the min entropy is higher when the size of this set is smaller. This would always produce the trivial bound of $\Hmin(A|E) \ge 0$ and so Theorem \ref{thm:twoparty} cannot be used for the three-party case. We prove that sampling can provide an entropic uncertainty relation in this scenario for high-dimensional states by suitably modifying the first step of our proof method.  Furthermore, we show how our proof method can lead to relations incorporating more than one overlap, useful in case the two bases have a shared vector in common (e.g., a ``vaccuum'' state vector for QKD).

\begin{theorem} \label{thm:three-party}
Let $\epsilon > 0$, $0 < \beta < 1/2$, and $\rho_{ABE}$ be an arbitrary quantum state acting on $\mathcal{H}_A\otimes\mathcal{H}_B\otimes\mathcal{H}_E$, where $\mathcal{H}_A\cong \mathcal{H}_B \cong \mathcal{H}_d^{\otimes (n+m)}$ with $d \ge 2$ and $m \le n$. Let $Z$ and $X$ be two orthonormal bases of $\mathcal{H}_{d}$ and define the maximal overlap $\hat{\gamma}$ as $\hat{\gamma}=-\log_2\max_{a,b}|\braket{z_a|x_b}|^2$.  Let $a^*$ and $b^*$ be a pair that attains this maximum, then we define the second-greatest overlap as:
\[
\gamma = -\log_2\max_{\substack{a\ne a^*\\b\ne b^*}}|\braket{z_a|x_b}|^2.
\]
(It is possible that $\gamma = \hat{\gamma}$ for some bases.) Let
\[
\delta = \sqrt{\frac{(m+n+2)\ln(4/\epsilon^2)}{m(m+n)}}.
\]
Finally, let $\rho_{TABE} = \frac{1}{T}\sum_t\kb{t}\otimes\rho_{ABE}$, where the sum is over all subsets of the form $t = (t_A,t_B)$ with $t_A = t_B$ (over their respective subspaces) and $T = {n+m \choose m}$. Then, except with probability at most $2\epsilon^{1-2\beta}$, after running $(t, q_A,q_B, \rho_{ABE}(t,q_A,q_B))\experiment{\rho_{TABE}, X}$, it holds that:
\[
    \Hmin^{4\epsilon+2\epsilon^\beta}(A_Z|E)_{\rho(t,q_A,q_B)} + \frac{n\Hextd( \hd(q_A,q_B) + \delta)}{\log_d2} \ge n(c_{b^*}(q_A) + \delta)\hat{\gamma} + n(1-c_{b^*}(q_A) - \delta)\gamma
\]
where $A_Z$ above, denotes the random variable resulting from measuring the remainder of the $A$ system of $\rho(t,q_A,q_B)$ in the $Z$ basis and the probability is over all choices of subsets and measurement outcomes within the experiment.  If $\hat{\gamma} = \gamma$, then the above simplifies to:
\[
    \Hmin^{4\epsilon+2\epsilon^\beta}(A_Z|E)_{\rho(t,q_A,q_B)} + \frac{n\Hextd( \hd(q_A,q_B) + \delta)}{\log_d2} \ge n\gamma
\]
\end{theorem}
\begin{proof}
As with our other proofs of sampling based entropic uncertainty relations, this one follows the same two-step structure where, first, we analyze the ideal case, proving the result there; then, finally, we argue that the real case must follow the ideal except with small probability of failure. For this three party version, only the first step changes from our proof of Theorem \ref{thm:twoparty}, the second step is identical.

Consider sampling strategy $\samp_{2+0}$ defined in Section \ref{section:sample-strats} with the count index set to $b^*$ which is the classical strategy we will employ in this scenario. By Theorem \ref{thm:sample}, there exist ideal states $\{\ket{\phi^t_{ABE}}\}$, indexed over all subsets $t=(t_A,t_B)$, such that $\ket{\phi^t_{ABE}}\in \text{span}(\mathcal{G}_{t,\delta})\otimes\mathcal{H}_E$ and, in this case as we are using $\samp_{2+0}$, the set
\[
\mathcal{G}_{t,\delta} = \{(i,j)\in\alphabet_d^N\times\alphabet_d^N \st |\hd(i_{t_A},j_{t_B}) - \hd(i_{-t_A},j_{-t_B})|\le \delta \text{ and } |c_{b^*}(i_{t_A}) - c_{b^*}(i_{-t_A})|\le \delta\}.
\]
Furthermore, by our choice of $\delta$, and the failure probability of $\samp_{2+0}$ (from Equation \ref{eq:samp-psi20}), we have:
$\frac{1}{2}\trd{\rho_{TABE} - \sigma_{TABE}} \le \epsilon,$ where $\sigma_{TABE}$ is the ideal state defined over all subsets and individual ideal states above (as in Theorem \ref{thm:sample}).
If we consider performing the given experiment on this ideal state, afterwards, we will receive as output the chosen subset $t$, the measurement results $q_A,q_B$, and the post-measurement state $\ket{\phi^t(q_A,q_B)}_{ABE}$ which is guaranteed to be of the form:
\[
\ket{\phi^t(q_A,q_B)} = \sum_{(i,j)\in J_{q_A,q_B}}\alpha_{i,j}\ket{i}_A\ket{j}_B\ket{E_{i,j}},
\]
where the above $\ket{i}_A$ and $\ket{j}_B$ are $X$ basis vectors (i.e., $\ket{x_i}$ and $\ket{x_j}$), and:
\[
J_{q_A,q_B} = \{(i,j)\in\mathcal{A}_d^{2n} \st |\hd(q_A,q_B) - \hd(i,j)| \le \delta \text{ and } |c_{b^*}(q_A) - c_{b^*}(i)|\le\delta\}.
\]

Rearranging terms and permuting the $A$ and $B$ subspaces, we may write the above state as:
\[
\ket{\phi^t(q_A,q_B)} \cong \ket{\widetilde{\phi}^t(q_A,q_B)} = \sum_{j\in Y} \widetilde{\alpha}_j\ket{j}_B\otimes \sum_{i\in J_{q_A,q_B}^{(j)}}\beta_i^{(j)}\ket{i}_A\ket{\widetilde{E}_{i,j}},
\]
where $Y \subset \alphabet_d^n$ and $J_{q_A,q_B}^{(j)}\subset \{i\in\mathcal{A}_d^n \st |\hd(i,j) - \hd(q_A,q_B)|\le\delta \text{ and } |c_{b^*}(q_A) - c_{b^*}(i)|\le\delta\}$.  Note that some of the $\widetilde{\alpha}$ and $\beta$'s may be zero.  Tracing out $B$ leaves us with:
\[
\sigma_{AE} = \sum_{j\in Y}|\widetilde{\alpha}_j|^2P\left[\sum_{i\in J_{q_A,q_B}^{(j)}}\beta_i^{(j)}\ket{i}_A\ket{\widetilde{E}_{i,j}}\right] = \sum_{j\in Y}|\widetilde{\alpha}_j|^2\sigma_{AE}^{(j)},
\]
where $P(z) = zz^*$. At this point, $A$ measures the remaining portion of her register in the $Z$ basis, resulting in $\sum_j\sigma_{A_Z,E}^{(j)}$. By appending a suitable classical system and conditioning on it, we may use Equation \ref{eq:mixed}, to show that
\[
\Hmin(A_Z|E)_\sigma \ge \min_j\Hmin(A_Z|E)_{\sigma^{(j)}}.
\]

Consider a particular $j$ and define $\chi^{(j)}_{AE} = \sum_{i\in J_{(q_A,q_B)}^{(j)}}|\beta_i^{(j)}|^2\kb{i}\otimes\kb{\widetilde{E_{i,j}}}$. From Lemma \ref{lemma:superposition}, we have:
\begin{align*}
\Hmin(A_Z|E)_{\sigma^{(j)}} &\ge \Hmin(A_Z|E)_{\chi^{(j)}} - \log_2|J_{(q_A,q_B)}^{(j)}|
\end{align*}
We first bound $\Hmin(A_Z|E)_{\chi^{(j)}}$.  Taking $\chi^{(j)}$ and measuring in the $Z$ basis yields:
\[
\chi^{(j)}_{ZE} = \sum_{i\in J_{(q_A,q_B)}^{(j)}} |\beta_i^{(j)}|^2\left( \sum_{z\in\alphabet_d^n}p(z|i)\kb{z}\right)\otimes\kb{\widetilde{E_{i,j}}}
\]
where
\[
p(z|i) = |\braket{z|x_i}|^2 = \prod_{k=1}^n|\braket{z_k|x_k}|^2
\]
We wish to find an upper bound on $p(z|i)$ for any $z$ and $i$ (within our constraints on $i$) which will be used shortly to bound the min entropy of the system.  Recall, we have two particular overlaps we are considering: one for $\braket{z_{a^*}|x_{b^*}}$ and one for the remaining possible pairs.  It is not difficult to see that $p(z|i)$ is maximized if, whenever $i_k = b^*$ that we have $z_k = a^*$.  This can happen at most $n(c_{b^*}(q_A) + \delta)$ times due to our constraint on $i$ and so the remaining counts (namely $n(1-c_{b^*}(q_A) - \delta)$) will be bounded using $\gamma$.  Thus, we conclude:
\begin{equation}
    p(z|i) = \prod_{k=1}^n|\braket{z_k|x_k}|^2 \le \left(|\braket{z_{a^*}|x_{b^*}}|^2\right)^{n(c_{b^*}(q_A) + \delta)}\times \left(\max_{\substack{a\ne a^*\\b\ne b^*}}|\braket{z_a|x_b}|^2\right)^{n(1-c_{b^*}(q_A)-\delta)}.
\end{equation}
Finally, we append a classical system spanned by orthonormal basis $\{\ket{i}_I\}$ for all $i\in J_{(q_A,q_B)}^{(j)}$ producing state:
\[\chi^{(j)}_{ZEI} = \sum_i |\beta_i^{(j)}|^2\left(\sum_z p(z|i)\kb{z}\right)\otimes\kb{\widetilde{E_{i,j}}}\otimes\kb{i}_I.\]
Then, using Equation \ref{eq:mixed} and the definition of min entropy, we conclude:
\begin{align*}
\Hmin(A_Z|E)_{\chi^{(j)}} &\ge \Hmin(A_Z|EI)_{\chi^{(j)}}\\
&\ge \min_i(-\log\max_z p(z|i))\\
&\ge n(c_{b^*}(q_A) + \delta)\hat{\gamma} + n(1-c_{b^*}(q_A) - \delta)\gamma.
\end{align*}

Finally, it is clear that:
\begin{align*}
    |J_{q_A,q_B}^{(j)}| &\le |\{i\in\alphabet_d^n\st |\hd(i,j) - \hd(q_A,q_B)|\le\delta\}\\
    &= |\{i\in\alphabet_d^n\st |\hd(i,0) - \hd(q_A,q_B)|\le\delta\}\\
    &\le d^{n\Hextd(\hd(q_A,q_B)+\delta)},
\end{align*}
where the last inequality follows from the well-known bound on the volume of a Hamming sphere. Since the above analysis holds for any $j$, we have therefore computed the resulting min entropy of the ideal case, namely for \emph{any} chosen $t$ and observed $q_A,q_B$, it holds that:
\begin{equation}
    \Hmin(A_Z|E)_\sigma \ge n\left((c_{b^*}(q_A) + \delta)\hat{\gamma} + (1-c_{b^*}(q_A) - \delta)\gamma - \frac{\Hextd(\hd(q_A,q_B)+\delta)}{\log_d 2}\right)
\end{equation}

The second step of the proof involves arguing that the smooth min entropy $\Hmin^{4\epsilon+2\epsilon^\beta}(A_Z|E)_\rho$, for the given input state $\rho_{ABE}$, is bounded by the same quantity with high probability. This can be done in the same way as the second step in Theorem \ref{thm:twoparty}.  Since the trace distance between the real and ideal states, for our chosen $\delta$, is no greater than $\epsilon$, the same error and smoothing bounds apply as in the second step in Theorem \ref{thm:twoparty}. thus completing the proof.

\end{proof}

\subsection{Application to QKD Security}

Entropic uncertainty relations involving three parties, $A$, $B$, and $E$ have numerous applications, especially in quantum cryptography. Here we demonstrate how our bound produces improved finite-key rate bounds for the High-Dimensional BB84 protocol (HD-BB84) introduced in \cite{HD-BB84}. High-dimensional QKD protocols have been shown to exhibit several advantages over qubit based protocols in some scenarios, including in noise tolerance. For a general survey of QKD protocols, the reader is referred to \cite{qkd-survey-new,qkd-survey} while for a survey specific to high-dimensional QKD, the reader is referred to \cite{HD-qkd-survey}.

HD-BB84 involves two orthonormal bases, which we denote $Z = \{\ket{0}, \cdots, \ket{d-1}\}$ and $X = \{\ket{x_0}, \cdots, \ket{x_{d-1}}\}$, each of dimension $d$; we will assume the bases are mutually unbiased and so $|\braket{i|x_j}| = 1/\sqrt{d}$ for all $i,j$.  If we are considering lossy channels, then we will also add a $\ket{vac}$ vector to both these bases.  Alice chooses a random basis and a random state within that basis (though not the $\ket{vac}$ state if it is there), sending it to $B$. $B$, on receipt of a quantum state will measure it in the $Z$ or $X$ basis, choosing randomly. Afterwards, a classical authenticated communication channel is used allowing $A$ and $B$ to inform each other of their basis choices. If they are incompatible, the round is discarded; otherwise, assuming $B$ did not observe $\ket{vac}$, they add $\log d$ bits to their \emph{raw key}. Repeating $N$ times, each $A$ and $B$ has a raw key of size $n$ bits. However, this key is only partially correlated (there may be errors due to natural noise or adversarial interference) and only partially secret. Thus, an Error Correction protocol is run (leaking additional information to the adversary) and, finally, Privacy Amplification (as discussed in Section \ref{section:notation}), resulting in a secret key of size $\ell$ bits. Maximizing $\ell$ is vital to efficient performance of QKD systems and, from Equation \ref{eq:PA}, this involves maximizing our estimate of the min entropy $\Hmin^\epsilon(A|E)$.

To analyze this protocol, we consider an equivalent entanglement based version, parameterized by $Z$, $X$, $n$ and $m$. We also consider an asymmetric version whereby only $Z$ basis measurements contribute to the raw key, while $X$ basis measurements are used only for estimating the error in the channel.
The entanglement based HD-BB84 runs as follows:
\begin{enumerate}
    \item An adversary prepares a quantum state $\ket{\psi_0} \in \mathcal{H}_A\otimes\mathcal{H}_B\otimes\mathcal{H}_E$, where $\mathcal{H}_A \cong \mathcal{H}_B \cong \mathcal{H}_d^{\otimes n+m}$. The $A$ portion is sent to Alice; the $B$ portion is sent to Bob; while Eve keeps the $E$ portion to herself.
    \item $A$ chooses a random subset $t$ of size $m$ and sends it to $B$; both parties measure their systems indexed by $t$ in the $X$ basis resulting in outcomes $q_A$ and $q_B$ respectively (these are strings in $\alphabet_d^m$). These values are disclosed to one another using the authenticated channel.
    \item $A$ and $B$ measure the remaining portion of their systems in the $Z$ basis resulting in their raw-keys $r_A$ and $r_B$ of size at most $n$ bits each (if there are $\ket{vac}$ observations, those will not contribute to the raw key and so it may be smaller than $n$ in a lossy channel).
    \item $A$ and $B$ run an error correction protocol capable of correcting up to $Q$ errors in their raw keys, leaking $\leakec$ bits to Eve.
    \item Finally, privacy amplification is run on the error corrected raw key resulting in their secret key.
\end{enumerate}

Note that when $d = 2$ this is exactly the BB84 protocol.  Note also that, by increasing the basis dimension to $d+1$, we can add an additional ``vacuum'' state $\ket{vac}$ to both the $Z$ and $X$ basis, such that $\braket{i|vac} = \braket{x_i|vac} = 0$.  In this case the maximal overlap function is $\hat{\gamma} = -\log_2 1 = 0$ and the second maximal overlap function is $\gamma = -\log_2 1/d = \log_2 d$.  (Note that this shows the importance of our relation in being able to handle both cases individually.) Without this vacuum basis state, the dimension will be $d$, and $\hat{\gamma} = \gamma = \log_2 d$.

Using Equation \ref{eq:PA} and results in \cite{renner2008security,tomamichel2012tight}, if $A$ and $B$ wish to have an $\epsilon_{PA}$-secure key, we have:
\[
\ell = \Hmin^{\epsilon'}(A|E) - \leakec - 2\log\frac{1}{\epsilon_{PA} - 2\epsilon'}.
\]
Given $\epsilon > 0$ and using our Theorem \ref{thm:three-party}, setting $\epsilon_{PA} = 4\epsilon^\beta + 9\epsilon$, we have:
\begin{equation}\label{eq:hd-bb84-our-rate}
    \ell_{our-HD-BB84} = n(1-p_{vac}-\delta)\left(\log d - \frac{\Hextd(\hd(q_A,q_B) + \delta)}{\log_{d+1} 2}\right) - \leakec - 2\log\frac{1}{\epsilon}
\end{equation}
where $p_{vac}$ is the number of counts in the observed $q_A$ of the distinguished vacuum basis state (which is shared between both the $Z$ and $X$ basis making $\hat{\gamma} = 0$). In particular, if the privacy amplification function is chosen to produce an output of size $\ell_{ours}$, it is guaranteed, except with probability at most $2\epsilon^{1-2\beta}$, that the secret key will be $\epsilon_{PA}$ secure according to Equation \ref{eq:PA}.  Note that if we are not considering lossy channels, then the key-rate equation becomes simply:
\begin{equation}\label{eq:hd-bb84-our-rate-2}
    \ell_{our-HD-BB84-no-loss} = n\left(\log d - \frac{\Hextd(\hd(q_A,q_B) + \delta)}{\log_{d} 2}\right) - \leakec - 2\log\frac{1}{\epsilon}
\end{equation}

To compare our new key-rate bound with prior work, we compare with results in \cite{HD-BB84-finite} which is, to our knowledge, the current best bound for the HD-BB84 protocol in the finite key setting (with composable security, as is ours).  Note that they used an entropic uncertainty relation from \cite{tomamichel2011uncertainty}, resulting in a key-rate bound of:
\begin{equation}\label{eq:hd-bb84-prior}
    \ell_{prior-HD-BB84} = n[\log_2 d - h(Q+\nu) - (Q+\nu)\log_2 (d-1)],
\end{equation}
where:
\[
\nu = \sqrt{\frac{(n+m)(m+1)\ln(2/\epsilon)}{m^2n}}.
\]
Where, for our evaluations, $Q$ is the error parameter of a depolarization channel.
Note that this prior work could not handle an additional vacuum basis state in each of the $Z$ and $X$ basis (if it were added, the bound from \cite{tomamichel2011uncertainty} would become the trivial one as the overlap function would be $-\log_2 1 = 0$).  So, when we evaluate, we will compare our bounds both without the vacuum basis then later by considering this basis state and loss in the channel.


In practice, the value of $\hd(q_A,q_B)$ or $Q$ is known and observed based on the actual channel used. However, to evaluate and compare our new key-rate bound we will 
evaluate assuming a depolarization channel with parameter $Q$ acting on each qudit independently and identically. Such a channel maps a quantum state $\rho$ to:
\[
\mathcal{E}_Q(\rho) = \left(1-\frac{d}{d-1}\cdot Q\right)\rho + \frac{Q}{d-1}I.
\]
Of course, our security proof does not require this depolarization assumption - instead, it is simply a channel we use to evaluate our bound and compare with prior work. It is also one of the most common noise models considered in theoretical QKD security proofs.  For both protocols, we use $\leakec = 1.2H(A|B)$ which, for this depolarization channel, is easily found to be $H(A|B) = Q\log(d-1) + h(Q)$.

\begin{figure}
    \centering
    \includegraphics[width=.46\textwidth]{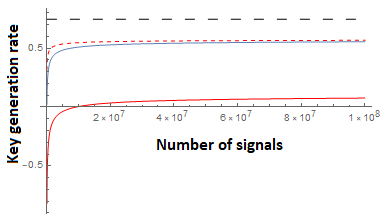}
    \includegraphics[width=.46\textwidth]{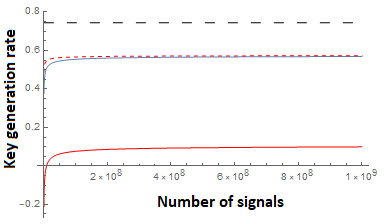}
    \caption{\footnotesize{Showing the secret key generation rates ($\ell/N$) of the HD-BB84 protocol when dimension $d=2^2$ assuming a depolarization channel with parameter $Q=10\%$. Here, the $x$-axis is the total number of qudits $N$ from which we use $m = .07N$ for sampling. Left and Right are different ranges in the number of signals.  Dashed black line (top most in both graphs) is the theoretical asymptotic rate (Equation \ref{eq:hd-bb84-asym}); Solid blue line is our key-rate bound using our new entropic uncertainty relation, namely $\ell_{our-HD-BB84-no-loss}/N$ (Equation \ref{eq:hd-bb84-our-rate-2}) for $p_{vac} = 0$ (no loss); Dashed red line is the previous best known bound for the HD-BB84 key rate using alternative methods to compute $E$'s uncertainty, $\ell_{prior-HD-BB84}/N$ (Equation \ref{eq:hd-bb84-prior}) with no loss (loss is not supported in that prior work); Finally, solid-red line (lowest) is our key-rate bound when $p_{vac} = 20\%$ (i.e., a $20\%$ loss in the channel) using Equation \ref{eq:hd-bb84-our-rate}. For our key-rate evaluation, we use $\beta = 1/3$ and $\epsilon = 10^{-36}$ giving a failure probability and a value of $\epsilon_{PA}$ both on the order of $10^{-12}$. For Equation \ref{eq:hd-bb84-prior}, we use a failure probability of $10^{-12}$. For both finite key results, we use $\leakec = 1.2H(A|B)$ which, in the case of a depolarization channel, is $\leakec = 1.2(Q\log(d-1) + h(Q))$. For the theoretical upper-bound we use the $\leakec = H(A|B)$ (without the additional $1.2$ scaling factor).}}
    \label{fig:HD-BB84-dim-exp2}
\end{figure}

\begin{figure}
    \centering
    \includegraphics[width=.46\textwidth]{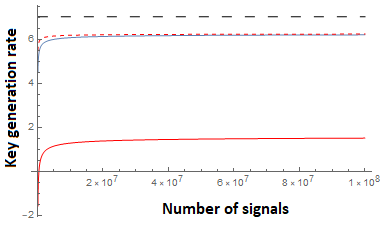}
    \includegraphics[width=.46\textwidth]{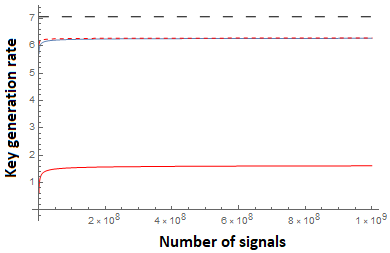}
    \caption{\footnotesize{Similar to Figure \ref{fig:HD-BB84-dim-exp2} but now showing the secret key generation rates ($\ell/N$) of the HD-BB84 protocol when $d=2^{10}$. Here the depolarization noise is $Q=10\%$. Dashed black line (top most in both graphs) is the theoretical asymptotic rate (Equation \ref{eq:hd-bb84-asym}); Solid blue line is our key-rate bound with no loss ($p_{vac} = 0$); Dashed red line is the previous best known bound for the HD-BB84 key rate (with no loss); finally, solid red line (lowest) is our key-rate bound when $p_{vac} = 50\%$.}}
    \label{fig:HD-BB84-dim-exp10}
\end{figure}

Finally, we compare to the theoretical, asymptotic upper-bound using the entropic uncertainty relation of \cite{berta2010uncertainty}. This disregards all finite-key effects (such as failure probabilities and sampling imprecision), and takes the number of signals $N \rightarrow \infty$. This bound works out easily to be:
\begin{equation}\label{eq:hd-bb84-asym}
r_{asym} = \log d - 2H(A|B) = \log d-2(Q\log(d-1) + h(Q)),
\end{equation}
where again we used the easily verified fact that, for a depolarization channel with parameter $Q$, $H(A|B) = Q\log(d-1) + h(Q)$ and, furthermore, we assume perfect error correction whereby $\leakec = H(A|B)$.

Comparisons of both our new bound and prior work are shown in Figure \ref{fig:HD-BB84-dim-exp2} (for $d=2^2$ dimensions) and Figure \ref{fig:HD-BB84-dim-exp10} (for dimension $d=2^{10}$).  We note that when $p_{vac} = 0$, our bound is only slightly lower than Equation \ref{eq:hd-bb84-prior} and this difference decreases as the number of signals increases.  Indeed, the difference turns out to be only that our confidence interval, determined by $\delta$ is slightly larger for any particular $\epsilon$ making our results asymptotically the same, though slightly lower than prior work for this case.  However, one of the powers of our new relation is its ability to also handle two overlap functions allowing us to incorporate loss in both $Z$ and $X$ bases.  Of course, as the loss increases, the key-rate decreases as expected; our new entropic uncertainty relation can, however, easily handle this scenario.  Further refinements to the classical sampling strategy used, may further improve our bound (in both the lossy and loss-less case).  Indeed our analysis of Lemma \ref{lemma:samp-psi2} is not necessarily tight. Alternative sampling strategies or improved analyses, may be easily incorporated through our methods.


\section{Closing Remarks}

The quantum sampling framework of Bouman and Fehr, introduced in \cite{bouman2010sampling}, provides a promising new tool to develop results in general quantum information theory and quantum cryptography.  In our prior work \cite{krawec2019quantum,krawec2020new}, we used this framework to introduce so-called sampling-based entropic uncertainty relations.  In this paper, we showed how quantum sampling can be used to develop very general quantum entropic uncertainty relations allowing one to insert arbitrary classical sampling strategies, perhaps defined for a specific cryptographic task, which may then be ``promoted'' to analyze results for quantum systems.  Furthermore, we developed an entirely new three-party entropic uncertainty relation using the sampling framework as a foundation, which has applications to high-dimensional QKD as we demonstrated here. Our new relation can also handle two different measurement overlaps, allowing one to work with bases that share common vectors (such as a ``vacuum'' measurement outcome).  Since our relation handles all finite sampling precision, they provide an easy and general purpose framework for other researchers to develop finite-key cryptographic security proofs.


Several interesting future problems remain open.  So far we only considered projective basis measurements.  Generalizing these results to arbitrary POVM's would be greatly interesting.  However, this would require extending the quantum sampling technique to support such measurements.  Furthermore, improving the three-party relation with a tighter sampling strategy would produce even more beneficial results.  Finding other interesting theoretical and cryptographic applications of quantum sampling and our sampling-based entropic uncertainty relations would also be highly interesting.  We feel that the framework of quantum sampling is powerful and can be employed successfully in other areas of quantum information science, and further exploration of quantum sampling in the domain of quantum information theory can yield even more exciting results in quantum cryptography.

\bibliography{qrng-qkd}
\bibliographystyle{unsrt}

\end{document}